\documentclass[11pt,english,notitlepage]{article}

\usepackage[utf8]{inputenc}
\usepackage[T1]{fontenc}
\usepackage{lmodern}
\usepackage[DIV=11]{typearea} 
\usepackage[margin=1in]{geometry}

\usepackage{microtype}

\usepackage[usenames,dvipsnames,svgnames,table]{xcolor}
\usepackage{enumitem}
\usepackage{caption}
\usepackage{placeins}
\usepackage{graphicx}
\usepackage{comment}
\usepackage{thm-restate}

\usepackage{fancybox}
\usepackage{hyperref}

\usepackage[normalem]{ulem}

\DeclareUnicodeCharacter{221A}{$\sqrt{}$}
\DeclareUnicodeCharacter{2113}{$\ell$}

\usepackage{amsmath,amsfonts,amsthm,amssymb,xcolor}
\usepackage{bm} 

\usepackage{nicefrac}

\usepackage[ruled,vlined, linesnumbered]{algorithm2e}

\usepackage{float}
\usepackage{cleveref}

\newcommand{\vol}{\text{vol}}
\newcommand{\BB}{\boldsymbol{B}}
\newcommand{\br}{\boldsymbol{r}}
\newcommand{\bgamma}{\boldsymbol{\gamma}}
\newcommand{\bflow}{\boldsymbol{f}}
\newcommand{\bc}{\boldsymbol{c}}
\newcommand{\bDeg}{\boldsymbol{deg_G}}
\newcommand{\bDelta}{\boldsymbol{\Delta}}
\newcommand{\bT}{\boldsymbol{\nabla}}
\newcommand{\bNull}{\boldsymbol{0}}
\newcommand{\bOnes}{\boldsymbol{1}}
\newcommand{\ex}{\boldsymbol{ex}_{\bflow}}
\newcommand{\abs}{\boldsymbol{abs}_{\bflow}}

\newtheorem{theorem}{Theorem}[section]
\newtheorem{corollary}[theorem]{Corollary}
\newtheorem{lemma}[theorem]{Lemma}

\newtheorem{claim}[theorem]{Claim}
\newtheorem{invariant}[theorem]{Invariant}
\newtheorem{fact}[theorem]{Fact}

\newtheorem{Informal Theorem}[theorem]{Informal Theorem}

\theoremstyle{definition}
\newtheorem{definition}[theorem]{Definition}
\newtheorem{remark}[theorem]{Remark}

\newtheorem*{theorem*}{Theorem}
\newtheorem*{corollary*}{Corollary}
\newtheorem*{conjecture*}{Conjecture}
\newtheorem*{lemma*}{Lemma}
\newtheorem*{thm*}{Theorem}
\newtheorem*{prop*}{Proposition}
\newtheorem*{obs*}{Observation}
\newtheorem*{definition*}{Definition}

\newtheorem*{remark*}{Remark}
\newtheorem*{rec*}{Recommendation}

\title{Maintaining Expander Decompositions via Sparse Cuts}
\author{Yiding Hua  \\ ETH Zurich \\ yidhua@student.ethz.ch \and Rasmus Kyng\thanks{The research leading to these results has received funding from the grant ``Algorithms and complexity for high-accuracy flows and convex optimization'' (no. 200021 204787) of the Swiss National Science Foundation.} \\ ETH Zurich  \\ kyng@inf.ethz.ch \and Maximilian Probst Gutenberg\footnotemark[1]\\ ETH Zurich \\ maxprobst@ethz.ch \and Zihang Wu \\ ETH Zurich \\ zihangwu98@gmail.com}

\date{}

\begin{document}
\pagenumbering{gobble}

\maketitle

\begin{abstract}
In this article, we show that the algorithm of maintaining expander decompositions in graphs undergoing edge deletions directly by removing sparse cuts repeatedly can be made efficient.

Formally, for an $m$-edge undirected graph $G$, we say a cut $(S, \overline{S})$ is $\phi$-sparse if $|E_G(S, \overline{S})| < \phi \cdot \min\{\vol_G(S), \vol_G(\overline{S})\}$. A $\phi$-expander decomposition of $G$ is a partition of $V$ into sets $X_1, X_2, \ldots, X_k$ such that each cluster $G[X_i]$ contains no $\phi$-sparse cut (meaning it is a $\phi$-expander) with $\tilde{O}(\phi m)$ edges crossing between clusters. A natural way to compute a $\phi$-expander decomposition is to decompose clusters by $\phi$-sparse cuts until no such cut is contained in any cluster. We show that even in graphs undergoing edge deletions, a slight relaxation of this meta-algorithm can be implemented efficiently with amortized update time $m^{o(1)}/\phi^2$. 

Our approach naturally extends to maintaining directed $\phi$-expander decompositions and $\phi$-expander hierarchies and thus gives a unifying framework while having simpler proofs than previous state-of-the-art work. In all settings, our algorithm matches the run-times of previous algorithms up to subpolynomial factors. Moreover, our algorithm provides stronger guarantees for $\phi$-expander decompositions.
For example, for graphs undergoing edge deletions, our approach is the first to maintain a dynamic expander decomposition where each updated decomposition is a refinement of the previous decomposition, and our approach is the first to guarantee a sublinear $\phi m^{1+o(1)}$ bound on the total number of edges that cross between clusters across the entire sequence of dynamic updates.
Our techniques also give by far the simplest, deterministic algorithms for maintaining Strongly-Connected Components (SCCs) in directed graphs undergoing edge deletions, and for maintaining connectivity in undirected fully-dynamic graphs, both matching the current state-of-the art run-times up to subpolynomial factors. 
\end{abstract}


\clearpage
\pagenumbering{arabic}

\section{Introduction}

During the last two decades, expanders and expander decompositions have been central to the enormous progress on fundamental graph problems. 

In static graphs expander decompositions were a fundamental tool to obtain the first near-linear time Laplacian solvers \cite{spielman2004nearly} and were used in many recent algorithms for maximum flow and min-cost flow problems \cite{kelner2014almost, van2020bipartite, van2021minimum, bernstein2021deterministic}. This, ultimately, led to an almost-linear time max flow and min-cost flow algorithm \cite{chen2022maximum} which crucially relies on techniques to maintain expanders undergoing edge deletions. Further, expanders have been central to all deterministic almost-linear time global min-cut algorithms for undirected graphs \cite{kawarabayashi2018deterministic, saranurak2021simple, li2021deterministic}, to compute short-cycle decompositions \cite{chu2020graph, parter2019optimal, liu2019short}, to find min-cut preserving vertex sparsifiers \cite{chalermsook2021vertex, liu2020vertex}, and have found many, many more applications.

In dynamic graphs, i.e.  graphs that are undergoing edge insertions and deletions over time, expanders played an equally important role in recent years. There, they have been behind new worst-case update time and derandomization results in dynamic connectivity \cite{wulff2017fully, nanongkai2017dynamic1, nanongkai2017dynamic, chuzhoy2020deterministic}, strongly-connected components \cite{bernstein2020deterministic}, single-source shortest paths  \cite{chuzhoy2019new, bernstein2020deterministic, chuzhoy2021deterministic, chuzhoy2021decremental, bernstein2021deterministic}, approximate $(s,t)$-max-flow and min-cut algorithms \cite{goranci2021expander}, and sparsifiers against adaptive adversaries \cite{bernstein2020fully}. They were also a key ingredient in the first subpolynomial update time $c$-edge connectivity algorithm \cite{jin2022fully}.

Given the enormous impact that expander techniques have had on the current state-of-the-art of graph algorithms, we therefore believe that it is important to further our understanding of expander maintenance. In this article, we give a new approach that goes well beyond previous techniques and that we believe is simple and accessible, works well in many settings (in directed graphs or graphs undergoing vertex splits, and so on), and even obtains stronger properties than previous algorithms. 
Concretely, our approach is the first to maintain an expander decomposition where each updated decomposition is a refinement of the previous one, and the first to achieve a sublinear bound $\phi m^{1+o(1)}$ on the total number of edges that cross between partitions summed across the entire sequence of updates.
We also show various interesting applications of our new techniques for many of the problems mentioned above, leading to simpler algorithms overall.

\subsection{Expanders and Expander Decompositions}

To advance the discussion let us formally define expanders. As expanders are objects closely related to flows, we let $G$ generally denote a \emph{directed}, unweighted multi-graph. We say $G$ is \emph{undirected}, if there is a one-to-one correspondence between edges $(u,v) \in E$ and $(v,u) \in E$. We let the degree of a vertex $v$ in $G$, denoted by $\deg_G(v)$, be the number of incident edges, i.e. edges with $v$ as tail or head. We define $\vol_G(X)$ for $X \subseteq V$ to be the sum of degrees, i.e. $\vol_G(X) = \sum_{v \in X} \deg_G(v)$. We let $E(A,B)$ for $A, B \subseteq V$ denote the edges in $E$ with tail in $A$ and head in $B$. We let $\overleftarrow{G}$ denote the graph $G$ with edges reversed, $G[X]$ be the graph induced by vertices in $X$, and let $G / X$ be the graph $G$ after contracting the vertices in $X$ into a single super-vertex. We say that a cut $(S, \overline{S})$ is $\phi$-out-sparse if $\vol_G(S) \leq \vol_G(\overline{S})$ and $|E(S, \overline{S})| < \phi \cdot vol_G(S)$ and $\phi$-sparse if it is $\phi$-out-sparse in $G$ or $\overleftarrow{G}$. This allows us to define the notion of expanders. 

\begin{definition}[Expander]
For any $\phi \in (0,1)$, we say that $G$ is a \emph{$\phi$(-out)-expander} if it has no $\phi$(-out)-sparse cut. 
\end{definition}

It is straight-forward to see that for undirected graphs, if $G$ is a $\phi$-out-expander, then it also is a $\phi$-expander, as we have symmetry in the cuts. Given the definition of an expander, we can define the following decomposition which is the central object of this article.

\begin{definition}[Expander Decomposition]
Given a directed graph $G$ and parameters $\phi \in (0,1],
\beta \geq 1$, we say that a tuple $(\mathcal{X}, E_{Rest})$ forms an $(\beta, \phi)$-expander decomposition of $G$ where $\mathcal{X}$ is a partition of $V$ and $E_{Rest} \subseteq E$ if 
(1) for each $X \in \mathcal{X}$, cluster $G[X]$ is a $\phi$-expander, and (2) $E_{Rest}$ is of size at most $\beta \phi m$, and (3) $G / \{X_i\}_i \setminus E_{Rest}$ is a DAG.
\end{definition}

We sometimes call $\beta$ the \emph{quality} of the expander decomposition. Note that for undirected graphs, we can extend the above set $E_{Rest}$ to always include the anti-parallel edge $(u,v)$ if already $(v,u)$ in $E_{Rest}$, and thus only loose a factor of $2$ in the size of $E_{Rest}$, but then obtain the property that $G / \{X_i\}_i \setminus E_{Rest}$ is a graph containing only self-loops. Put differently, $E_{Rest}$ contains all edges between clusters.

\begin{definition}[Undirected Expander Decomposition]
Given an \emph{undirected} graph $G$ and parameters $\phi \in (0,1],
\beta \geq 1$, we say that a tuple $(\mathcal{X}, E_{Rest})$ forms an $(\beta, \phi)$-expander decomposition of $G$ where $\mathcal{X}$ is a partition of $V$ and $E_{Rest} \subseteq E$ if 
(1) for each $X \in \mathcal{X}$, cluster $G[X]$ is a $\phi$-expander,
and (2) $E_{Rest}$ is the set of edges not in any cluster and is of size at most $2\beta \phi m$.
\end{definition}


\subsection{A Natural Meta-Algorithm for Expander Decomposition} 

To obtain a $(\tilde{O}(1), \phi)$-expander decomposition, the following meta-algorithm is folklore.

\begin{algorithm}
$\mathcal{X} \gets \{V\}$; $E_{Rest} \gets \emptyset$.\;
\While{there is a $\phi$-out-sparse cut $(S, X \setminus S)$ in $G[X]$ or $\overleftarrow{G}[X]$ for $X \in \mathcal{X}$}{
    Replace $X$ in $\mathcal{X}$ by sets $S$ and $X \setminus S$.\;
    Add to $E_{Rest}$ the smaller set of edges $E_{G[X]}(S, X \setminus S)$ or $E_{G[X]}(X \setminus S, S)$.\label{lne:addSparseCutInMeta}
}
\Return $(\mathcal{X}, E_{Rest})$
\caption{$\textsc{MetaAlgorithm}(G,\phi)$}
\label{alg:metaAlg}
\end{algorithm}

Let us analyze this meta-algorithm. To see that the while-loop terminates, it suffices to observe that each while-loop iteration decomposes a set $X \in \mathcal{X}$ further and thus after $n-1$ iterations, each set in $\mathcal{X}$ is a singleton set $\{v\}$ for some vertex $v \in V$. But $G[\{v\}]$ forms a trivial $\phi$-expander. Let us next argue that  there are at most $\tilde{O}(\phi m)$ edges in $E_{Rest}$ by the end of the algorithm: every time $E_{G[X]}(S, X \setminus S)$ or $E_{G[X]}(X \setminus S, S)$ is added in \Cref{lne:addSparseCutInMeta}, the number of added edges to $E_{Rest}$ is at most $\phi \vol_{G[X]}(S) \leq \phi \vol_G(S)$. But each vertex $s \in S$ is contained in a cluster with at most half the number of edges compared to the cluster $G[X]$. Thus, each vertex $v \in V$ can be at most $O(\log(m))$ times on the smaller side of the sparse cut. This implies our bound. It remains to use the condition of the while-loop to conclude that the output of the algorithm is indeed a $(\tilde{O}(1), \phi)$-expander decomposition.

\paragraph{Implementing the Meta-Algorithm Efficiently.} We point out that since finding a $O(1)$-approximate $\phi$-out-sparse cut even in an undirected graph is NP-hard \cite{chawla2006hardness} under the Unique Games Conjecture, any polynomial time implementation of the meta-algorithm has to resort to relaxing the algorithm to taking \emph{approximate} sparsest cuts. 

The first implementation of this relaxed meta-algorithm was already given in \cite{kannan2004clusterings} where expander decompositions were proposed. However, their straight-forward use of a static procedure to find a $\tilde{O}(\phi)$-sparse cut in each while-loop iteration caused them a $\Omega(mn)$ run-time since each iteration might only find a very unbalanced $\tilde{O}(\phi)$-sparse cut, leading to recursion depth of $\Omega(n)$ in the worst case.

Later, near-linear time algorithms were found that implement the meta-algorithm more loosely. The first such work in undirected graphs was by Spielman and Teng \cite{spielman2004nearly} who proposed spectral local methods to locate balanced $O(\sqrt{\phi})$-sparse cuts which allowed them to obtain \emph{near-expanders} (a weaker notion of expanders). 

The framework by Nanongkai, Saranurak and Wulff-Nilsen \cite{wulff2017fully, nanongkai2017dynamic1, nanongkai2017dynamic} finally gave the first efficient implementation of the meta-algorithm that only used (relatively) balanced $n^{o(1)}\phi$-sparse cuts, resulting in total time $m^{1+o(1)}$ to compute an expander decomposition in undirected graphs. A key ingredient in their work was the adaption of flow based techniques to obtain improved approximation guarantees over the framework of Spielman and Teng. These flow techniques were in turn pioneered in \cite{khandekar2009graph, peng2016approximate, orecchia2014flow}. The framework in \cite{bernstein2020deterministic} further extended this technique to directed graphs with similar guarantees (up to subpolynomial factors).

Recently, Saranurak and Wang \cite{saranurak2019expander} also gave an algorithm to compute $(\tilde{O}(1), \phi)$-expander decompositions in undirected graphs in time $\tilde{O}(m/\phi)$ which improves on the above runtime for the important case where $\phi = \tilde{\Omega}(1)$. The algorithm can be seen as an even further refinement of the flow based techniques in \cite{wulff2017fully, nanongkai2017dynamic1, nanongkai2017dynamic} obtaining almost optimal approximation guarantees and run-time. We point out however that this algorithm relaxes the above meta-algorithm even further by also using non-sparse cuts when convenient. 

\paragraph{The Meta-Algorithm for Dynamic Graphs.} Interestingly, the meta-algorithm is also natural for graphs undergoing edge deletions. More precisely, a natural way to extend the meta-algorithm is to run its while-loop after each edge deletion on the clusters given from before the deletion. The same analysis from before can now be made to conclude that even after $m$ deletions, the maximum number of edges to ever join the set $E_{Rest}$ (which is now a monotonically increasing set) is at most $\tilde{O}(\phi m)$. In fact, the above analysis even holds for graphs $G$ undergoing $\tilde{O}(m)$ edge deletions, vertex splits and self-loop insertions.

The main contribution of this article is to show that the meta-algorithm can even be implemented efficiently for graphs undergoing edge deletions, vertex splits and self-loop insertions (although at the additional cost of $m^{o(1)}$ in the sparsest cut approximation). This starkly differs from previous algorithms to maintain dynamic expander decompositions \cite{henzinger2020local,nanongkai2017dynamic1, nanongkai2017dynamic, wulff2017fully, bernstein2020deterministic, chuzhoy2020deterministic, goranci2021expander} which all take non-sparse cuts and have to maintain $E_{Rest}$ as a fully-dynamic set to retain reasonable size where $E_{Rest}$ has to undergo up to $\tilde{O}(m)$ total changes. This strengthening of properties on the expander decomposition then allows us to give a unified theorem that combines various previous results while losing at most subpolynomial factors in quality and run-time of the algorithm. We give a formal statement of our contribution in the next section and an overview of techniques in \Cref{subsec:techiques}.

\subsection{Our Contributions}

We summarize our main result in the Theorem below where the Theorem works for both directed and undirected graphs even though the definitions of expander decompositions differ slightly in these settings.

\begin{restatable}{theorem}{expanderMain}[Randomized Dynamic Expander Decomposition]\label{thm:mainTheoremRand}   \label{thm:mainTheoremDet} Given an $m$-edge graph $G$ undergoing a sequence of $\tilde{O}(m)$ updates consisting of edge deletions, vertex splits and self-loop insertions, parameters $\phi \in (0,1)$ and $1 \leq L_{max} = O(\sqrt{\log\log m})$.

Then, we can maintain a $(\gamma, \phi/\gamma)$-expander decomposition $\mathcal{X}$ for $\gamma = (\log(m))^{4^{O(L_{max})}}$ with the properties that at any stage (1) the current partition $\mathcal{X}$ is a refinement of all its earlier versions, and (2) the set $E_{Rest}$ is a super-set of all its earlier versions. The algorithm implements the meta-algorithm in \Cref{alg:metaAlg} and takes total time $\tilde{O}(m^{1+1/L_{max}}\gamma/\phi^2)$ and succeeds with high probability.
\end{restatable}

In the theorem above, the algorithm works against an adaptive adversary, i.e. the adversary can design the update sequence to $G$ on-the-go and based on the previous output. \Cref{thm:mainTheoremRand} can also be derandomized by replacing a randomized subroutine with a deterministic counterpart (as was presented in \cite{bernstein2020deterministic}). This comes however at the cost of increasing $\gamma$ slightly. Still, for some appropriate choice of $L_{max}$, the algorithm maintains a $(m^{o(1)}, \phi/m^{o(1)})$-expander decomposition in time $m^{1+o(1)}/\phi^2$. If vertex splits are disallowed from the update sequence, then the runtime can be improved to $m^{1+o(1)}/\phi + m^{o(1)} t/\phi^2$ where $t$ is the number of updates.

We point out that this matches previous state-of-the-art algorithms \cite{henzinger2020local, wulff2017fully, nanongkai2017dynamic1, nanongkai2017dynamic,saranurak2019expander, chuzhoy2020deterministic, bernstein2020deterministic} to maintain $\phi$-expander decompositions up to a subpolynomial factor in quality and run-time in every setting (i.e. even for the special case of allowing randomization and considering only undirected, simple graphs undergoing only edge deletions).

Interestingly, $\phi$-expander hierarchies as introduced in \cite{goranci2021expander} can also be maintained straight-forwardly using the Theorem above (see Application \#2 in \Cref{subsec:applications}).

\subsection{Techniques}\label{subsec:techiques}

We now give an overview of our techniques. To simplify matters, we present our new algorithm only for \emph{directed} graphs $G$ undergoing edge deletions.

\paragraph{High-level Approach.} The key ingredient to our algorithm is the maintenance of a witness graph $W$ for each expander graph $G$. Intuitively, $W$ is a graph that is easier to work with and that can be used as an explicit certificate that $G$ is an expander.

When $G$ undergoes a set of edge deletions $D$, it turns out that we can leverage our knowledge of $W$ to detect potential sparse cuts in $G \setminus D$. Moreover, setting up flow problems carefully, we can then check if one of the potential sparse cuts is indeed a real sparse cut. If so, we return the sparse cut. Otherwise, we can find a new witness graph $W'$. 

In contrast to our algorithm, previous approaches to expander maintenance did not use witnesses, but rather tried to locate sparse cuts in $G$ directly.
This however came at the loss of not being able to locate the real sparse cuts but rather previous algorithms could only identify a subgraph $(G \setminus D)[X]$ that is still expander (for $X$ being rather large) but could not make more fine-grained statements.

In the next paragraphs, we define what a witness graph is, then explain how to maintain witnesses of $\phi$-expanders that are affected by a large number of deletions and finally sketch how to use such witness maintenance to achieve \Cref{thm:mainTheoremRand}.

\paragraph{Expanders via Witness Graphs.} It is well-known in the literature that given a $\phi$-expander $G$, one can find a $\psi$-expander $W$ over the same vertex set as $G$ such that $\psi = \Omega(1/\log^2(m))$ and degree vectors $\mathbf{deg}_{W} \approx \bDeg$, along with a routing $\Pi_{W \mapsto G}$ such that for each edge $e = (u,v) \in E(W)$, $\Pi_{W \mapsto G}(e)$ maps to a $u$ to $v$ path in $G$; with the additional property that $\Pi_{W \mapsto G}$ has \emph{congestion} at most $\frac{1}{\phi\psi}$ meaning that no edge in $G$ appears on more than $\frac{1}{\phi\psi}$ such paths. In fact, for $G$ being $\phi$-expander, the algorithms in \cite{khandekar2009graph, louis2010cut} compute such a witness $W$ and routing $\Pi_{W \mapsto G}$ in time $\tilde{O}(m/\phi)$, w.h.p. even in directed graphs. We point out that $W$ is also a directed graph.

Given such a graph $W$ and routing $\Pi_{W \mapsto G}$, it is straight-forward to prove that $G$ must be a $\Omega(\phi\psi^2) = \tilde{\Omega}(\phi)$-expander (see \Cref{clm:conductanceForRNull}). Therefore $W$ is often called the \emph{witness graph}. 

\paragraph{Maintaining the Witness Graph of a $\phi$-Expander.} In our approach, we are maintaining a witness graph for each expander graph. The main ingredient towards maintaining the witness graph, is to handle a (large) batch of updates to the expander graph and recover a witness. We call the act of handling these deletions \emph{one-shot pruning}. We give the following Informal Theorem which is made formal in \Cref{sec:one-shotPruning}.

\begin{Informal Theorem}
\label{ref:informalTheoremSloppy}
Given a directed graph $G$, a $\psi$-expander witness $W$ (where $\psi$ as above) over the same vertex set with routing $\Pi_{W \mapsto G}$ of congestion $\frac{1}{\phi\psi}$, a set of edges $D \subseteq E$ with $|D| \ll \phi|E|$, $\deg_{G \setminus D}(v) \approx \deg_{G}(v)$ for all $v \in V$. 

Then there is an algorithm $\textsc{PruneOrRepair}$ that either
\begin{itemize}
    \item returns a $\tilde{\Omega}(\phi)$-sparse cut $(S, V \setminus S)$ in $G \setminus D$, or
    \item returns a new  $\Omega(\psi^3)$-expander $W'$ and embedding $\Pi_{W' \mapsto G \setminus D}$ with congestion $\tilde{O}(\frac{1}{\phi})$ (and therefore certifies that $G\setminus D$ is still $\tilde{\Omega}(\phi)$-expander).
\end{itemize}
The algorithm runs in time $\tilde{O}(|D|/\phi^2)$.
\end{Informal Theorem}

Here, the rather strange-looking assumption that $\deg_{G \setminus D}(v) \approx \deg_{G}(v)$ is purely to simplify the presentation below and can be removed entirely.

Our approach to maintaining the witness is straightforward: when an edge $e_G$ is added to $D$ (and hence deleted in $G\setminus D$), we remove each edge $e_W$ of $W$ that was routed through the deleted edge $e_G$ in the embedding $\Pi_{W \mapsto G}$. To repair the witness, we will attempt to add new edges in $W$,
leaving the endpoints of each deleted edge $e_W$. The heads (starting point) of these new edges will be at the endpoints of edges removed from $W$, but the tails (endpoints) may be at different nodes.  We call the repaired witness $W'$.
For technical reasons, we attempt to add a few more edges to the new witness $W'$ than we deleted from the old witness $W$.
However, before adding these edges, we first want to make sure we can embed them into the updated graph $G$ without too much additional congestion. 
To certify that the new edges of $W'$ are embeddable into with little congestion, we introduce a flow problem whose solution will either let us embed these new witness edges into $G$, or find a sparse cut in $G$.

Importantly, we will be able to use a local algorithm to solve the flow problem on $G$, i.e. we do not need to explore the entire graph, but can instead run an algorithm that only visits a small part of $G$ in the neighborhood of $D$. This is essential to establishing our running time.

\paragraph{Cuts or Witness Maintenance via Flow.}
We  set up a flow problem that lets us implement the witness repair or one-shot pruning described above. 
The flow problem asks us to route flow in the graph $G$.
The flow demands we seek to route are guided by the deletions to $G$, 
and chosen to help us add edges to repair our witness $W$ whenever witness edges embedded into $G$ have been impacted by a deletion in $G$.

In the following paragraphs, we set parameters to match, up to polylogarithmic factors, the parameters in the rest of the article but often simplify by omitting constants since we are relying on assumptions that are not properly quantified in the overview (for example that $\deg_{G \setminus D}(v) \approx \deg_{G}(v)$). We do so to keep the overview intuitive and to avoid overly technical details.

Consider the following flow algorithm on the graph $G \setminus D$. Let $\bDelta \in \mathbb{N}^{V}$ be the amount of flow that has to be routed away from vertices in $V$ (i.e. the source vector). Initially, we set $\bDelta$ to be the all-zero vector. Then for each $e = (u,v) \in D$, we find the edges $e' = (x,y) \in \Pi^{-1}_{W \mapsto G}(D)$, i.e. the edges $e'$ such that $e \in \Pi_{W \mapsto G}(e')$, and place $8/\psi$ units of demand at both vertices $x$ and $y$. The figure below illustrates such a case where in the left graph, the embedding path $\Pi_{W \mapsto G}(e')$ is drawn and can be seen to use the edge $(u,v) = e \in D$.

\begin{figure}[h]
\centering
\includegraphics[width=0.7\textwidth]{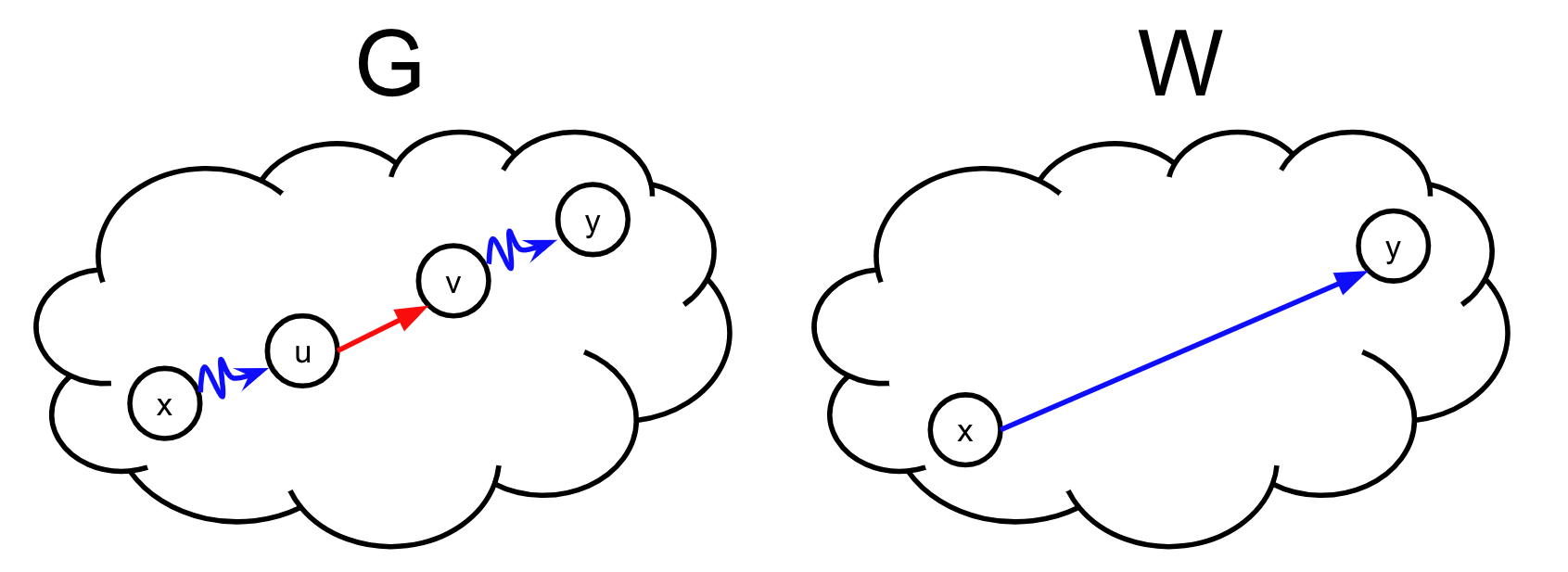}
\end{figure}

We then set-up a sink vector $\bT \in \mathbb{N}^V$ that we set equal to the degree vector $\mathbf{deg}_{G \setminus D}$ of the graph $G \setminus D$. Finally, we define a capacity vector $\bc = \frac{1}{\psi^2\phi} \cdot \bOnes \in \mathbb{R}^{E \setminus D}$ and then try to find a flow $\bflow \in \mathbb{N}^{E \setminus D}$ that sends the maximum amount of source flow to the sinks while respecting the capacities. This can be done using a max-flow algorithm (we use a modification of the blocking flow algorithm which provides similar guarantees as used below). 
We use some basic combinatorial properties of the blocking flow algorithm and our flow problem to ensure the algorithm runs locally, visiting only a small neighborhood around $D$.
We point out that by the assumption $|D| \ll \phi|E|$, we make sure that the flow problem is a diffusion problem, i.e. that $\|\bT\|_1 \geq \|\bDelta\|_1$.

\paragraph{Finding a Sparse Cut (If Source Flow is not Routed).} If $\bflow$ cannot route all the flow away from the sources, or more formally, if there is a vertex $v$ with $[\BB^\top \bflow + \bDelta](v) > \bT(v)$ where $\BB$ is the incident matrix of $G \setminus D$, then we claim the algorithm can extract a $\tilde{O}(\phi)$-sparse cut. 

To see this, let $(S, \overline{S})$ be the min-cut in the flow network. By the max-flow min-cut theorem, we have that the total capacity of edges from $S$ to $\overline{S}$ must be smaller than the total source demand $\bDelta$ on $S$:
\[
    c(E_{G \setminus D}(S, \overline{S})) < \bDelta(S).
\]
By our choice of capacities, this immediately gives that:
\[
    |E_{G \setminus D}(S, \overline{S})| < \psi^2\phi \bDelta(S) = \tilde{O}(\phi \bDelta(S)).
\]
Thus, if we can show that $\bDelta(S) \leq \tilde{O}(\vol_G(S))$, then we can conclude that $(S, \overline{S})$ is indeed a $\tilde{O}(\phi)$-sparse cut (here we implicitly assumed $\vol_G(S) \leq \vol_G(\overline{S})$).

To this end, we recall that $\deg_{G \setminus D}(v) \approx \deg_G(v) \approx \deg_W(v)$ for all vertices $v \in V$. But note that the way we constructed $\bDelta(v)$ is by placing $8/\psi = \tilde{O}(1)$ units on $v$ for each edge incident to $v$ in $W$ that was removed in our procedure. But since $\deg_W(v) \approx \deg_G(v)$, we thus get our desired bound.

\paragraph{Repairing the Witness (If Source Flow is Routed).} If $\BB^\top \bflow + \bDelta \leq \bT$, then the algorithm can use $\bflow$ to \emph{repair} the witness $W$ to obtain a new witness $W'$. Therefore, it initializes $W' = W \setminus \Pi^{-1}(D)$. Then, it runs a path-decomposition algorithm on $\bflow$ and for each $x$ to $y$ path in the decomposition, we add a new edge $(x,y)$ to $W'$. 

Note that this also induces a natural routing $\Pi_{W' \mapsto G}$ by routing along the underlying flow path for each new edge in $W' \setminus W$. It is further not hard to observe that the congestion of $\Pi_{W' \mapsto G}$ is at most the congestion of $\Pi_{W \mapsto G}$ plus an additive term of $\frac{1}{\psi^2\phi}$ which stems from the capacity in the flow problem which upper bounds the number of flow paths routed through the edge.

To verify that $\mathbf{deg}_{G \setminus D} \approx \mathbf{deg}_{W'}$, we can simply use our assumption that $\mathbf{deg}_{G \setminus D} \approx \mathbf{deg}_{G}$ and the fact that for each edge incident to vertex $v$ in $W$ that was in $\Pi^{-1}_{W \mapsto G}(D)$, we place $\tilde{\Theta}(1)$ units of source flow which then translates to new edges with $v$ as its tail (since we can route $\bflow$) while on the other hand, by setting $\bT = \mathbf{deg}_{G \setminus D}$, we ensure that there are at most $\mathbf{deg}_{G \setminus D}(v)$ new edges with head in $v$ in $W'$.

Finally, we prove that for each cut $(S, \overline{S})$ where 
$\vol_{W'}(S) \leq \vol_{W'}(\overline{S})$, we have $|E_{W'}(S, \overline{S})| = \Omega(\frac{1}{\psi^3}) \vol_{W'}(S)$. We point out that this will only show that $W'$ is a $\Omega(\frac{1}{\psi})$-\emph{out}-expander instead of showing that it is an expander. However, by applying the same algorithm to the graphs $G$ and $W'$ with edges reversed, we can recover and show that either a sparse cuts from this procedure is found or a graph $W''$ is found that is both out- and in-expander and therefore expander.

We prove the claim on the expansion of $(S, \overline{S})$ by a simple case analysis (see \Cref{fig:proofIllus} for an illustration of this proof):
\begin{itemize}
    \item If at least half the edge from $E_W(S, \overline{S})$ are also in $E_{W'}(S, \overline{S})$: then the claim follows immediately as this implies
    \[
    |E_{W'}(S, \overline{S})| \geq \frac{1}{2}|E_{W}(S, \overline{S})| \geq \frac{\psi}{2}\vol_W(S) \approx \frac{\psi}{2}\vol_{W'}(S).
    \]
    
    \item Otherwise: then it is not hard to verify that each of the edges that were removed from $W$ from the cut $E_W(S, \overline{S})$ adds $8/\psi$ units of source demand on a vertex in $S$, and therefore $\bDelta(S) \geq 4|E_W(S, \overline{S})|/\psi$. 
    
    We can now use that $|E_W(S, \overline{S})|/\psi \geq \vol_W(S) \approx \vol_{G}(S) \approx \vol_{G \setminus D}(S)$. Thus, we can upper bound the amount of flow that $S$ can absorb by $\bT(S) = \vol_{G \setminus D}(S) \lesssim |E_W(S, \overline{S})|/\psi$. Since the flow $\bflow$ was routed, that means that at least $|E_W(S, \overline{S})|/\psi \approx \vol_{W'}(S)$ units of source demand on $S$ were routed to vertices in $\overline{S}$ and subsequently each such unit of flow added one edge $(x,y)$ where $x \in S, y \in \overline{S}$ to $W'$. 
    
    Thus, we have that $|E_{W'}(S, \overline{S})| \gtrsim \psi \vol_{W'}(S)$.
\end{itemize}
  \begin{figure}[h]
    \centering
    \includegraphics[width=0.9\textwidth]{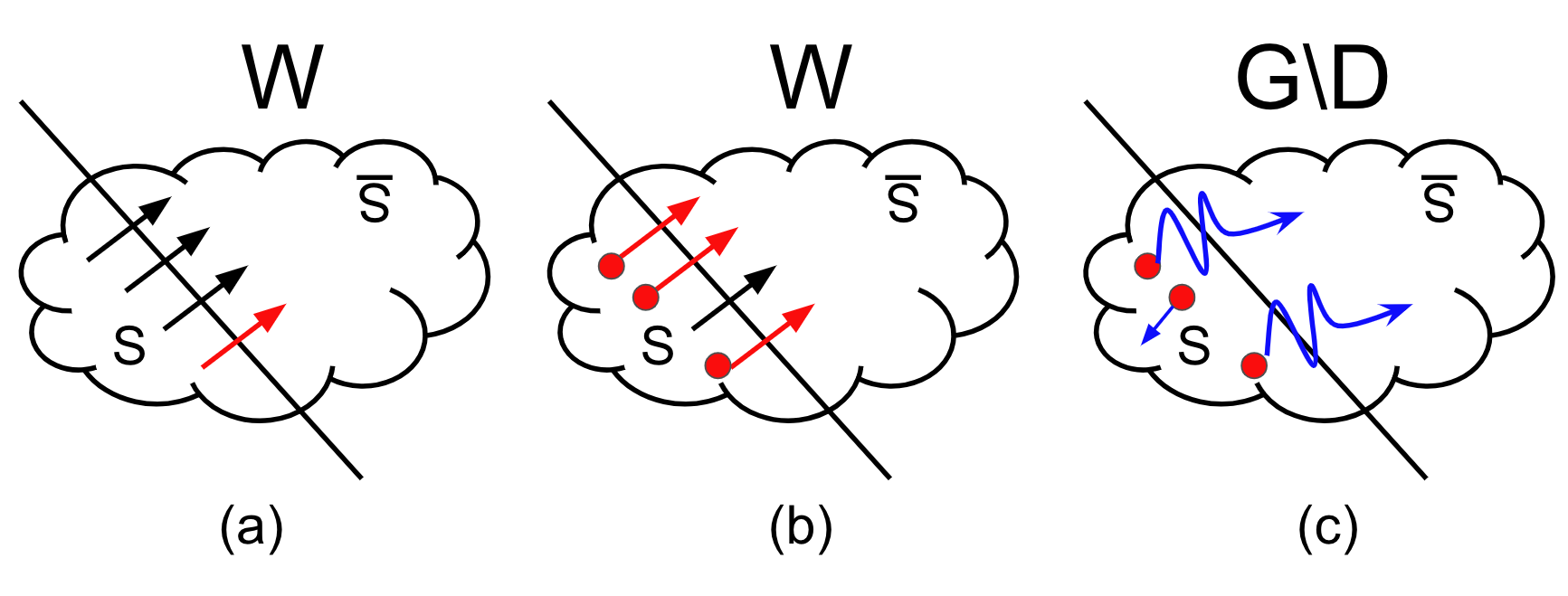}
    \caption{(a) and (b) show the old witness graph $W$ and the edges in $W$ crossing the cut from $S$ to $\overline{S}$. The red edges are the edges in $W$ that do not appear in $W'$. (a) corresponds to the first case in our proof and shows that if most edges survive, the expansion of $W'$ is still sufficient. In the other case, depicted in (b), each edge in $W \setminus W'$ contributes some demand on $S$ (depicted by the red points). Then, in the graph $G \setminus D$ given in (c), the demand on $S$ has to be routed and since it exceeds the sink capacity of $S$, most source flow is routed to vertices in $\overline{S}$. Each such flow path from $S$ to $\overline{S}$ is then be converted into a new edge in $W'$ that crosses the cut.}
    \label{fig:proofIllus}
    \end{figure}
        
\paragraph{From One-Shot Pruning to Expander Decomposition Maintenance via Batching.} Finally, the reader might wonder how to obtain an algorithm to maintain directed expander decompositions from the above one-shot pruning algorithm. At a high level, our algorithm maintains an expander decomposition $\mathcal{X} = \{X_1, X_2, \dots, X_{\tau}\}$ for graph $G$ by invoking \emph{one-shot pruning} upon batches of updates. This batching technique was developed in \cite{wulff2017fully, nanongkai2017dynamic1, nanongkai2017dynamic} and was derived from standard techniques in dynamic algorithms.

In order to make our one-shot pruning work efficiently in this setting, we first have to make it more resilient: a key problem with one-shot pruning in its current form is that it could return a very small sparse cut $(S, \overline{S})$ (i.e. one where $\vol_G(S) \gg \vol_G(\overline{S})$), then prompting us to recurse on almost the same problem again since we want to arrive at some $\overline{S}$ that is indeed expander again. Thus, we extend our one-shot pruning algorithm to always either output a \emph{large} sparse cut or certify that there is no large sparse cut in the witness. The Informal Theorem below makes this more explicit. It is a parameterized (in $R'$) version of the Informal Theorem \ref{ref:informalTheoremSloppy} where changes are colored blue.

\begin{Informal Theorem}
\label{ref:inftheoremParameterized}
Given a directed graph $G$, a $\psi$-expander witness $W$ (where $\psi$ as above) over the same vertex set with routing $\Pi_{W \mapsto G}$ of congestion $\frac{1}{\phi\psi}$, a set of edges $D \subseteq E$ with $|D| \ll \phi|E|$, $\deg_{G \setminus D}(v) \approx \deg_{G}(v)$ for all $v \in V$ {\color{blue} and a parameter $R' \geq 0$}. 

Then there is an algorithm $\textsc{PruneOrRepair}$ that either
\begin{itemize}
    \item returns a $\tilde{\Omega}(\phi)$-sparse cut $(S, V \setminus S)$ in $G \setminus D$ with {\color{blue}$\min\{\vol_{G \setminus D}(S), \vol_{G \setminus D}(V \setminus S)\} \geq R'$}, or
    \item returns {\color{blue}a new set $D' \subseteq D$}, a $\Omega(\psi^3)$-expander $W'$ and embedding {\color{blue}$\Pi_{W' \mapsto G \setminus (D \setminus D')}$} with congestion $\tilde{O}(\frac{1}{\phi})$ (and therefore certifies that {\color{blue}$G \setminus (D \setminus D')$} is still $\tilde{\Omega}(\phi)$-expander) {\color{blue}such that $|D'| \leq R'$}. 
\end{itemize}
The algorithm runs in time $\tilde{O}(|D|/\phi^2)$.
\end{Informal Theorem}

Using this refined version of the Informal Theorem, we can now implement the approach laid out above efficiently by recursing on a witness with no large sparse cuts where the threshold for being ``large'' scales in the depth of the recursion level. In our article, we start the recursion at some large level $L_{max} + 1$ for some appropriately chosen value $L_{max} = O(\sqrt{\log\log(m)})$ and go down in level with each level of recursion until we reach level $0$.

\paragraph{Formal Set-Up of the Batch-Update Framework.}
For the sake of concreteness, we now give a more formal description of the algorithm (without specifying all details yet as this is rather tedious and does not necessarily help intuition). Our algorithm maintains on each level $l = L_{max} +1, L_{max}, \ldots, 0$, for each $X \in \mathcal{X}$, a witness graph $W_{X, l}$. We say that a level $l$ is recomputed whenever $W_{X,l}$ is changed. We additionally maintain the sets $D_{X,l}$ of adversarial deletions since the last time level $l$ was recomputed, and sets $D'_{X, l}$ to capture adversarial deletions not handled during recomputation. 

Initially, we use a static routine to compute a $\phi_0$-expander decomposition and set all witnesses $W_{X,l}$ to the corresponding witness, and set all sets $D_{X,l}$ and $D'_{X,l}$ to be empty. Throughout, we maintain the invariant that each graph $W_{X,l}$ is a witness that $G[X] \cup D_{X,l} \cup D'_{X,l}$ is a $\phi_l$-expander where we choose $\phi \approx  \phi_{L_{max} + 1} \gtrsim  \phi_{L_{max}} \gtrsim \ldots \gtrsim \phi_0 \approx \phi$ and that $|D_{X,l} \cup D'_{X,l}| < m^{l/L_{max}}$. Note that this implies that $W_{X,0}$ proves that $G[X]$ is a $\phi_0 \approx \phi$-expander as desired. So, if the invariant holds after each update, we indeed correctly maintain an expander decomposition.

Let us now describe how we process an adversarial edge deletion. For each deletion to $G[X]$, we add the deleted edge to all sets $D_{X, L_{max}+1}, D_{X, L_{max}}, \ldots, D_{X, 0}$. We then search for the largest $l$ where our invariant on the size of $D_{X, l} \cup D'_{X, l}$ is violated. Whenever this is the case, we invoke our Informal Theorem 
\ref{ref:inftheoremParameterized} on $W_{X, l-1}$ with the set of edges $D_{X, l} \cup D'_{X, l-1}$ and parameter $R' = \frac{1}{2}m^{l/L_{max}}$. 

Let us first consider the case that the algorithm never returns a sparse cut. Then, we have that Informal Theorem \ref{ref:inftheoremParameterized} returns a set $D'$ and a new witness $W'$. We set $W_{X,l}$ to $W'$, set $D_{X,l} = \emptyset$ and $D'_{X,l} = D'$. Then, we recompute in the same way for all lower levels $l-1, l-2, \ldots, 0$ where the invariant is violated. Note that since the witness $W'$ is slightly worse in quality than the original witness $W$, we need to choose $\phi_{l}$ to be slightly worse than $\phi_{l-1}$. Using the approach above, we have that for a level $l$, we only need to recompute the witness roughly every $\frac{1}{2}m^{l/L_{max}}$ adversarial deletions. That is since after each recomputation, the set $D'_{X,l}$ is far from violating the size invariant and $D_{X,l}$ is empty and thus, it takes many adversarial updates (or a higher level update which happens very, very few times) until it violates the invariant again.

On the other hand, when we find a sparse cut with smaller side of volume $\tilde{\Omega}(m^{l/L_{max}})$ as promised by the Informal Theorem for our choice of $R'$ at level $l$, we make a lot of progress. In fact, it is not hard to see that we only pay $\tilde{O}(m^{(l+1)/L_{max}}/\phi^2)$ time to find such a sparse cut by using the fact that the invariant holds for level $l+1$.

We refrain from formalizing this approach even further here and refer the interested reader to 
\Cref{sec:maintainBatching}.

\paragraph{Dealing with Vertex Splits.} In the rest of this paper, because we deal with vertex splits (and edge insertions), maintaining the sets $D_{X,l}$ and $D'_{X,l}$ would not capture all update types and we would need additional sets to capture vertex splits and edge insertions. Also Informal Theorem \ref{ref:inftheoremParameterized} would be difficult to state in a clean way. We therefore introduce $\br_{X,l}$ vectors as a handy representation to unify update types where $\br_{X,l}$ lives in $\mathbb{N}_{\geq 0}^{X}$. 

To understand how we use the $\br_{X,l}$ vectors, let us first describe what we do in case of a deletion. Observe that the crucial piece of information about the deletion of an edge $(u,v)$ in the process of repairing the witness (or finding a cut) is to find the edges $(x, y) \in W_{X,l}$ that use $(u,v)$, and then set up a flow problem adding flow to endpoints $x$ and $y$. 

We suggest to store that information directly by adding one unit to each vertex $x$ and $y$ for each such edge $(x, y)$ to the corresponding vertices in $\br_{X,l}$. Thus, we increase the $\ell_1$ sum of $\br_{X,l}$ by $\tilde{O}(1/\phi)$ for each edge deletion. We deal with edge insertions of an edge $(u,v)$ by simply adding one unit to $\br_{X,l}$ in the components $u$ and $v$. Finally, when splitting a vertex $v$ into $v'$ and $v''$ (where $v''$ has smaller volume), then for each edge $(x,y)$ that embeds into an edge that is now incident to $v''$, we add a unit to $\br_{X,l}$ to the vertices $x$ and $y$. The increase in the $\ell_1$ sum of $\br_{X,l}$ is only increased by $\tilde{O}(\deg_G(v'')/\phi)$ by this operation. This processing of updates to update $\br_{X,l}$ vectors can be directly seen in \Cref{alg:updateAlgo}.

Using this representation, we can write clean statements and unify proofs about various update types.

\paragraph{A Subtle Issue in Directed Graphs.} Finally, when turning to directed graphs, we want to make the reader aware of a rather subtle issues that makes proofs rather finicky. To overcome the issue we introduce $\bgamma_{X,l}$ vectors for each $X \in \mathcal{X}$ and level $l$.
The need for these vectors essentially arises from the following detail: recall that we define a cut $(S, \overline{S})$ to be $\phi$-out-sparse iff $\vol_G(S) \leq \vol_G(\overline{S})$ and $|E_G(S, \overline{S})| < \phi \cdot \vol_G(S)$. But in directed graphs, it turns out to be more useful to detect cuts $(S, \overline{S})$ that are sparse relative to a set of weights that differ slightly from the original vertex degrees that are used to define $\vol_G(S)$.
We specify these weights using a vector $\bgamma_{X,l}$, and denote the weight of a set $S$ by $\bgamma_{X,l}(S)$ (the sum of the weights of vertices in $S$).
We then look for generalized sparse cuts where $\bgamma_{X,l}(S) \leq \bgamma_{X,l}(\overline{S})$ and $|E_G(S, \overline{S})| < \phi \cdot \vol_G(S)$.
This is because after some deletions $D$, we might have that $\vol_G(S) \leq \vol_G(\overline{S})$ but $\vol_{G \setminus D}(S) > \vol_{G \setminus D}(\overline{S})$. Thus, we would have to show that $|E_G(S, \overline{S})| < \phi \cdot \vol_G(\overline{S})$. While the above asymmetry does not cause any problems in the undirected graph problem, it causes problems when we move to directed graphs. This is because we run one version of the algorithm in Informal Theorem \ref{ref:inftheoremParameterized} to find $\phi$-out-sparse cuts and another to find $\phi$-in-sparse cuts. Unfortunately, the above batching can cause the algorithms to be invoked on slightly different sets $D'_1$ and $D'_2$ of previously deleted edges from higher levels. Then, the issue above can mean that some $\phi$-out-sparse or $\phi$-in-sparse cuts are not detected properly. By using the $\bgamma_{X,l}$ vectors we can keep cuts fixed in direction. For the vectors $\bgamma_{X,l}$, we can use the original degree vectors in $G[X]$ which allows us to roughly recover the real $\phi$-out and $\phi$-in-sparse cuts as long as the total amount of volume in $G[X]$ is not changed by a constant fraction due to the deletions $D$. This generalizes seamlessly to insertions and vertex splits.

\subsection{Applications}\label{subsec:applications}

\paragraph{Application \#1: A Simple Proof of Decremental Strongly-Connected Components.} In the decremental strongly-connected components problem, the algorithm is given a decremental $m$-edge graph $G$, that is a graph that undergoes only edge deletions. The goal is to maintain the strongly-connected components (SCCs) in the graph $G$ explicitly over the entire update sequence. 

The currently best deterministic algorithm for this problem \cite{bernstein2020deterministic} obtains total update time $mn^{2/3 + o(1)}$. Here, we give an extremely simple algorithm that achieves $m^{1+2/3 + o(1)}$ total update time which matches the previous result for very sparse graphs. We note if randomization is allowed to solve the above problem, then a $\tilde{O}(m)$ algorithm is known \cite{bernstein2021decremental}. 

We first introduce the following proposition which was already used by previous algorithms for the problem (see \cite{chechik2016decremental,bernstein2020deterministic}). We point out that this proposition is obtained by a very simple and elegant algorithm itself and we encourage the interested reader to consult \cite{lkacki2013improved}.

\begin{theorem}[see \cite{lkacki2013improved}]
\label{thm:lkacki2013improved}
Consider an algorithm that for a decremental graph $G$ maintains a set $S \subseteq E(G)$ that is a super-set of its earlier versions at any time and maintains the SCCs in $G \setminus S$. Then, there is an algorithm that maintains the SCCs of $G$ in additional total time $\tilde{O}(m|S|)$. 
\end{theorem}

But this means for the graph $G$, we can pick $\phi = m^{-1/3}$ and maintain an $(m^{o(1)}, m^{-1/3 - o(1)})$-expander decomposition on $G$. We let $S = R$ and observe that the fact that expanders have no sparse cuts implies that $\mathcal{X}$ are exactly the connected components of $G \setminus S$. Further $|S| \leq m^{2/3+o(1)}$. The result follows. 

Our dynamic expander decomposition always produces refinements over time, and this exactly matches the requirements of \Cref{thm:lkacki2013improved}, leading to an extremely simple algorithm.
In contrast, because they lacked this guarantee,  \cite{bernstein2020deterministic} needed a much more elaborate approach (spanning over 80 pages):
they maintained for each SCC in $G\setminus S$ an almost-expander that was slowly decaying in size until it had to be reset. To then certify that a single SCC does not break into two pieces, the algorithm has to root decremental single-source shortest path data structures from the contracted set of vertices still in the almost-expanders. New vertices in $S$ were created when distances in an SCC became large which forced sparse cuts. The many moving parts above made this data structure complicated to describe, analyze and implement. 

We believe that the technique of congestion balancing from \cite{bernstein2020deterministic} can be applied to our framework rather straight-forwardly, yielding a slightly more complicated algorithm, but still vastly simpler algorithm, with total update time $mn^{2/3+o(1)}$.

\paragraph{Application \#2: A Simple Proof of the Expander Hierarchy.} We start by defining an expander hierarchy which were introduced in  \cite{goranci2021expander}. We point out that currently there is no sensible generalization of expander hierarchies to directed graphs, thus we let all graphs $G$ be undirected in this section.

We define $G[X]^{\tau}$ to be the graph $G$ induced by vertex set $X$ where each vertex $v \in X$ receives an additional number of $\lceil \tau \cdot |E(\{v\}, V \setminus X)| \rceil$ self-loops. For a partition $\mathcal{X}$, we denote by $G[\mathcal{X}]^{\tau}$ the union of the graphs $G[X]^{\tau}$ for $X \in \mathcal{X}$. This allows us to define Boundary-Linked expander decompositions.

\begin{definition*}[Undirected Boundary-Linked Expander Decomposition]
  Given an \emph{undirected} graph $G$ and parameters $\phi,\alpha \in
  (0,1], \beta > 0$,
  we say that a tuple $(\mathcal{X}, R)$ forms an $(\alpha,\beta,
  \phi)$-expander decomposition
  of $G$ where $\mathcal{X}$ is a partition of $V$ and $R \subseteq E$ if 
(1) for each $X \in \mathcal{X}$, cluster $G[X]^{\alpha/\phi}$ is a $\phi$-expander,
and (2) $R$ is the set of edges not in any cluster and is of size at most $\beta \phi m$.
\end{definition*}

We can derive the following Corollary from \Cref{thm:mainTheoremRand}, our main result.

\begin{theorem}[Undirected Boundary-Linked Dynamic Expander
  Decomposition]
  \label{thm:bndryLnkdDynED}
  Consider an $m$-edge undirected graph $G$ undergoing a sequence of $O(m)$ updates consisting of edge deletions, vertex splits and self-loop insertions, parameters $\phi \in (0,1)$ and $L_{max} \in \mathbb{N}_{\geq 1}$. 

Then, we can maintain a $(1/\gamma,\gamma, \phi/\gamma)$-expander decomposition
$\mathcal{X}$ for $\gamma = (\log(m))^{5^{O(L_{max})}}$ with the
properties that at any stage (1) the current partition $\mathcal{X}$ is
a refinement of all its earlier versions, and (2) the set $R$ is a
super-set of all its earlier versions. The algorithm takes total time
$m^{1+1/L_{max}}\gamma/\phi^2$ assuming that at most $\tilde{O}(m)$ self-loops are inserted over the course of the algorithm and succeeds with high probability.
\end{theorem}
\begin{proof}[Proof sketch.]
  We run the Dynamic Expander Decomposition of
  Theorem~\ref{thm:mainTheoremDet} on a copy of $G$ which we denote by $H$.
  When $G$ is changed by a dynamic update, we make the same change to
  $H$.
  But additionally, in $H$, we add self-loops to vertices, whenever an edge incident on the vertex
  enters $R$.
  We call these \emph{regularizing self-loops}.
  
  When a new edge $e$ enters $R$, for each of its endpoints $u$,
  let $R_u$ be the set of edges in $R$ incident of $u$, and let $s_u$
  be the number of regularizing self-loops placed on $u$ in so far.
  If $|R_u|\cdot \frac{1}{\phi}  > s_u$, we add additional
  regularizing self-loops to $u$ until this is no longer the case.
  Note that adding these self-loops may cause further changes to the
  partition, which may in turn cause additional edges to be added to
  $R$, and this may require us to add yet more regularizing
  self-loops. However, in a moment, we will argue that this process
  does not create too many cut edges.
  First, though, let us observe that for each $X \in \mathcal{X}$,
  the regularizing self-loops precisely ensures that $G[X]^{\frac{1}{\phi}} = H[X]$ is
  a $\frac{\phi}{\gamma}$-expander with boundary-linkedness
  parameter $\alpha = \frac{1}{\gamma}$.

  Finally, we need to argue that the addition of regularizing
  self-loops does not mean that we cut too many edges.
  The underlying guarantee of \Cref{thm:mainTheoremDet} ensures
  that starting with $m$ edges, after $O(m)$ updates, we cut at most
  $\phi m$ edges (i.e. put them into $R$).
  However, this implies that we add most
  $2\phi m \cdot \frac{1}{\phi}  = 2m$ regularizing self-loops.
  Thus, the addition of self-loops does not exceed our budget, and still
  leaves us able to receive further $O(m)$ self-loops updates to $G$,
  albeit with a slightly smaller budget.
\end{proof}

Next, we define expander hierarchies using terminology inspired by \cite{goranci2021expander} (but slightly adapted for convenience).

\begin{definition}[Undirected Dynamic Expander Hierarchy]
An $(\alpha, \beta, \phi)$-expander hierarchy is recursively
defined to consist of levels $0 \leq i \leq k$ where we have graphs
$G_i$ where $G_0 = G$ and an $(\alpha,\beta, \phi)$-expander decomposition
$\mathcal{X}_i$ of $G_i$ and we define $G_{i+1}$
recursively to be the graph $G$ after contracting the sets in the
expander decomposition $\mathcal{X}_{i+1}$ and removing self-loops, and finally have that
$G_k$ consists of only a single vertex.
\end{definition}

Finally, we can prove the main result of this section.

\begin{theorem}\label{thm:expanderHierarchy}
Consider an $m$-edge undirected graph $G$ undergoing a sequence of $\tilde{O}(m)$ updates consisting of edge deletions, vertex splits and edge insertions, parameters $\phi \in (0,1/4\gamma)$ and $L_{max} \in \mathbb{N}_{\geq 1}$. 
We can maintain a $(1/\gamma, 2\gamma, \phi/\gamma)$-expander hierarchy with $k =
  \log_{\frac{2\gamma}{\phi}}(m) + 1$ levels with total update time
  $O(m^{1+1/L_{max}}\gamma/\phi^2)$.
  The algorithm works against an adaptive adversary and succeeds with
  high probability.
\end{theorem}

\begin{proof}[Proof sketch.]
We construct the dynamic expander hierarchy as follows:
  
  We construct the dynamic
  expander decomposition of $G_i$ at each level using
  \Cref{thm:bndryLnkdDynED}.
  When a partition gets refined at level $i$, this corresponds to a
  vertex in level $i+1$ splitting into two, with edges resulting
  between. We can maintain the
  graph at level $i+1$ by first inserting self-loops on the vertex
  corresponding to the vertex which is about to split, and then
  splitting the vertex while turning the self-loops into edges between
  the vertices. For edge insertions to $G$, we simply add the inserted edge directly to the set $R$ of edges in between clusters in all graphs $G_i$ where it crosses. Whenever there are more than $(2\gamma \beta)^i m$ edges in the set $R$ at a graph $G_i$, we restart the expander decompositions on graph $G_i, G_{i+1}, \ldots, G_k$ via \Cref{thm:bndryLnkdDynED}.
  
  By the guarantees of this theorem, using induction on the level
  $i$, we can directly show that $G_i$ has at most $(2\gamma\phi)^i m$ edges
  initially, and receives at most $\tilde{O}(m)$
  updates, all of the forms allowed by \Cref{thm:bndryLnkdDynED}.

  At level $k = \log_{\gamma/\phi}(m)+1$, we have $<1$ edge left,
  and so the graph must be a single vertex.
  The running time for each time we run \Cref{thm:bndryLnkdDynED} at level $i$ is
  $O((\gamma\phi)^i m^{1+1/L_{max}}\gamma/\phi^2)$. It is not hard to verify that at level $i$, the algorithm in \Cref{thm:bndryLnkdDynED} is restarted at most $\tilde{O}(1/(\gamma \phi)^i)$ times. It remains to sum over the levels to obtain the run-time guarantees.
\end{proof}

\paragraph{Application \#3: Dynamic Connectivity with Subpolynomial Worst-Case Update/Query Times.} We can use the fact that we can derandomize \Cref{thm:expanderHierarchy} by using the deterministic version of \Cref{thm:mainTheoremRand}, and we can turn amortized update times in worst-case update times by using standard rebuilding techniques (see for example \cite{goranci2021expander}). Again, both changes come at the cost of increasing the constant $\gamma$, however, we can still find a $\phi = 1/m^{o(1)} \ll 1/\gamma$ such that the hierarchy runs with initialization time $m^{1+o(1)}$ and then processes each update in time $m^{o(1)}$. 

As discovered in \cite{goranci2021expander}, dynamic expander hierarchies immediately imply a simple dynamic connectivity algorithm. Since we have streamlined the implementation of dynamic expander hierarchies further, this gives an even simpler dynamic connectivity algorithm. 

For convenience, we describe this algorithm here: more precisely, we explain how to use the expander hierarchy to answer connectivity queries: for any two vertices $u,v \in V$, one can travel upwards in the hierarchy by going from the vertex $x$ in a graph $G_i$ to the vertex $y$ in $G_{i+1}$ where $x$ is in the expander that was contracted to obtain $y$. One can then compare the vertices that $u$ and $v$ reach in graph $G_k$ by traveling upwards repeatedly and if they are the same, $u$ and $v$ must be connected, otherwise they are not connected. The query can be implemented in $k = o(\log(m))$ time.

\section{Preliminaries}

\paragraph{Graphs.} In this article, we deal with directed, unweighted multi-graphs $G$. We let $E(G)$ denote the edge set of $G$ and $V(G)$ the vertex set. While technically in multi-graphs $G$, an edge $e \in E(G)$ cannot be encoded only by its endpoints, we commonly abuse notation and write $e = (u,v)$ to mean that $e$ is an edge with tail in $u$ and head in $v$. We let $\overleftarrow{G}$ denote the graph $G$ where edges are reversed.

\paragraph{Dynamic Graphs.} We consider dynamic graphs $G$, that is graph $G$ that undergo updates consisting of edge deletions and vertex splits. In the case of a vertex split of $v$, the adversary specifies the edges incident to $v$ that are moved to a new vertex $v'$ that is split from $v$. We assume that the adversary always specifies a vertex split update such that after the update the degree of $v'$ is at most the degree of $v$. Additionally, we allow for self-loop insertions.

\paragraph{Degree and Volume.} We define the degree $\deg_G(v)$ of a vertex $v \in V(G)$ to be the number of edges incident to $v$ where a self-loop counts $2$ units towards the volume of $v$. For any subset $S \subseteq V$, we define the volume $\vol_G(S) = \sum_{v \in S} \deg_G(v)$. 

\paragraph{Cuts.} When the context is clear, we define for a vertex subset $S$ in graph $G$, $\overline{S} = V(G) \setminus S$ and let $E_G(S, \overline{S})$ be the set of edges in $E(G)$ with tail in $S$ and head not in $S$. Given a vector $\br \in \mathbb{N}_{\geq 0}^{V(G)}$, we generalize the notions and say a cut $(S, \overline{S})$ where $\vol_G(S) + \br(S) \leq \vol_G(\overline{S}) + \br(\overline{S})$ is $(\br, \phi)$-out-sparse if $|E_G(S, \overline{S})| + \br(S) < \phi  (\vol_G(S) + \br(S))$ where $\br(S) = \sum_{s \in S} \br(s)$. When the vector $\br$ is not given explicitly, we assume $\br = \bNull$ (where $\bNull$ denotes the all-0 vectors) and also say a cut is $\phi$-out-sparse or $\phi$-in-sparse.

\paragraph{Expander.} We say that a graph $G$ and vector $\br \in \mathbb{N}_{\geq 0}^{V(G)}$ form an $(\br, \phi)$-out-expander if there is no $(\br, \phi)$-out-sparse cut. We say that $G$ is a $(\br, \phi)$-expander if both $G$ and $\overleftarrow{G}$ are $(\br, \phi)$-out-expander.

\paragraph{Embedding.} Given graphs $G$ and $W$ over the same vertex set. We say that a function $\Pi_{W \mapsto G}$ is an embedding of $W$ into $G$, if for each $e = (u,v) \in E(W)$, $\Pi_{W \mapsto G}(e)$ is a $u$-to-$v$ path in $G$. We let the inverse of an embedding, denoted $\Pi^{-1}_{W \mapsto G}$ map any set of edges $E' \subseteq E(G)$ to the set of edges in $E(W)$ whose embedding paths contain an edge in $E'$. We define the congestion of $\Pi_{W \mapsto G}$ by $cong(\Pi_{W \mapsto G}) = \max_{e \in E(G)} |\{ e' \in E(W) | e \in \Pi_{W \mapsto G}(e') \}|$. 

\paragraph{Witness.} To prove that a graph $G$ is an expander one can compute a well-known expander $W$ and embed it into $G$ with low congestion. Thus, $W$ is witnessing that $G$ is expander. Here, we generalize the concept slightly.

\begin{definition}[$R$-Witness]\label{def:alphaWitness}
Given a graph $G$, vectors $\br, \bgamma \in \mathbb{N}_{\geq 0}^{V(G)}$,  parameters $\phi, \psi \in (0,1)$ and $R \in \mathbb{N}_{\geq 0}$, we say that a graph $W$ over the same vertex set as $G$ along with an embedding $\Pi_{W \mapsto G}$ of $W$ into $G$ is an $(R,\phi,\psi)$-out-witness of $(G,\br)$ with respect to $\bgamma$ if
\begin{enumerate}
    \item $\| \br \|_1 \leq R$, and  \label{prop:alphaWitness1}
    \item we have $\deg_W(v) + \br(v) \in [\deg_{G}(v),\frac{1}{\psi}\deg_{G}(v)]$.\label{prop:alphaWitness4}
    \item for every cut $(S, \overline{S})$ with $\bgamma(S) \leq \bgamma(\overline{S})$, we have $|E_W(S, \overline{S})| + \br(S) \geq \psi(\vol_W(S) + \br(S))$, and
    \label{prop:alphaWitness2}

    \item $\Pi_{W \mapsto G}$ has congestion $\frac{1}{\psi\phi}$, and \label{prop:alphaWitness3}
\end{enumerate}
We say that $W$ is an $(R, \phi, \psi)$-witness of $(G,\br)$ with respect to $\bgamma$ if $W$ it is an $(R,\phi,\psi)$-out-witness of $G$ w.r.t. $\bgamma$ and $\overleftarrow{W}$ is an $(R,\phi,\psi)$-out-witness of $\overleftarrow{G}$ w.r.t. $\bgamma$.
\end{definition}

\begin{claim}\label{clm:conductanceForRGeneral}\label{clm:conductanceForRNull}
Given a graph $G$, if there exists a $(R, \phi, \psi)$-witness $W$ for $(G, \bNull)$ with respect to any $\bgamma$, then $G$ is a $\psi^2\phi$-expander.
\end{claim}
\begin{proof}
Given any cut $(S,\overline{S})$ where $\vol_G(S) \leq \vol_G(\overline{S})$. By \Cref{def:alphaWitness}, Properties \ref{prop:alphaWitness2} and \ref{prop:alphaWitness4}, we either have $\bgamma(S) \leq \bgamma(\overline{S})$ which implies $\min\{|E_{W}(S, \overline{S})|,|E_{W}(\overline{S}, S)|\} \geq \psi\vol_W(S) \geq \psi\vol_G(S)$; or we have that $\bgamma(S) > \bgamma(\overline{S})$ which implies $\min\{|E_{W}(S, \overline{S})|,|E_{W}(\overline{S}, S)|\} \geq \psi\vol_W(\overline{S}) \geq \psi\vol_G(\overline{S}) \geq \psi\vol_G(S)$.

That is, in either case, we can conclude $\min\{|E_{W}(S, \overline{S})|,|E_{W}(\overline{S}, S)|\} \geq \psi\vol_G(S)$. It remains to use Property \Cref{prop:alphaWitness3} to argue that $|E_{W}(S, \overline{S})| \leq \frac{1}{\phi\psi}|E_{G}(S, \overline{S})|$ since each edge in $G$ is used on at most $\frac{1}{\phi\psi}$ embedding paths of $\Pi_{W \mapsto G}$. The same argument holds for $|E_{W}(\overline{S}, S)|$ which completes the proof.
\end{proof}

We use the following result regarding the computation of witnesses. We use $\psi_{CMG}$ throughout the rest of the paper for a fixed input graph.

\begin{theorem}[see \cite{khandekar2009graph,louis2010cut,chuzhoy2020deterministic, bernstein2020deterministic}] \label{thm:cutMatching}
There is a randomized algorithm $\textsc{CutOrEmbed}(G, \phi, R)$ that given an $m$-edge graph $G$ and parameters $\phi \in (0,1), 0 \leq R$ outputs either
\begin{enumerate}
    \item a set $S \subseteq V$ where $R \leq \vol_G(S)$ with $|E_G(S, \overline{S})| < \phi \vol_G(S)$, or
    \item a vector $\br \in \mathbb{N}_{\geq 0}^{V(G)}$
    , and a graph $W$ and embedding $\Pi_{W \mapsto G}$ that form an $(R, \phi, \psi_{CMG})$-witness of $(G, \br)$ w.r.t. $\bgamma = \bDeg$ where $\psi_{CMG} = \Omega(1 / \log^2(m))$.
\end{enumerate}
The algorithm runs in time $\tilde{O}(m/\phi)$  and succeeds with probability at least $1-n^{-c}$ for any pre-fixed constant $c > 0$.
\end{theorem}

\paragraph{Flow.} A \emph{flow-problem} $\mathcal{I} = (G, \bc, \bDelta, \bT)$ consists of a graph $G$, with capacities $\bc \in \mathbb{N}_{\geq 0}^{E(G)}$, and source and sink functions $\bDelta, \bT \in \mathbb{N}_{\geq 0}^{E(V)}$. Letting $\BB$ be the incidence matrix of $G$. Then a vector $\bflow \in \mathbb{N}_{\geq 0}^{E(G)}$ is a \emph{pre-flow} if $\bNull \leq \bflow \leq \bc$ (entry-wise). Given a pre-flow $\bflow$ for a flow problem $\mathcal{I}$ as above, we define the \emph{flow absorption} vector $\abs = \min\{ \BB^\top \bflow + \bDelta, \bT\}$ to be the entry-wise minimum. We define the \emph{excess flow} $\ex =  \BB^\top \bflow + \bDelta - \abs, $. We say that $\bflow$ is an $R$-\emph{flow} if it is a pre-flow and additionally $\|\ex\|_1 \leq R$. Given a pre-flow $\bflow$, we define the residual graph $G_{\bflow}$ to be the graph obtained by adding for each edge $e = (u,v) \in E(G)$, an edge $\overrightarrow{e} = (u,v)$ to $G_{\bflow}$ of \emph{residual capacity} $\bc(e) - \bflow(e)$ and an edge $\overleftarrow{e} = (v,u)$ of capacity $\bflow(e)$. We let $\bc_{\bflow}$ be the residual capacities on the residual graph.

\paragraph{Misc.} We use $[k]$ to denote the set $\{0,1,\ldots, k\}$.

\section{One-Shot Pruning}
\label{sec:one-shotPruning}

The main result of this section is the following Lemma which either outputs a (large) sparse cut, or outputs a better witness. Note that the Lemma inputs a witness $W$ but can only output an out-witness $W'$ (we can remedy this by running the algorithm on the same parameters but with $G$ replaced by $\overleftarrow{G}$). 

\begin{lemma}\label{lma:asymPruning}
Given an $n$-vertex, $m$-edge graph $G$, vectors $\br, \bgamma \in \mathbb{N}_{\geq 0}^{V(G)}$, an $(R, \phi, \psi)$-witness $W$ of $(G, \br)$ w.r.t. $\bgamma$, for $\psi \leq \psi_{CMG}/8$, and a threshold $R' \in \mathbb{R}_{\geq 0}$ such that $R' \leq R \leq \psi m/8$. Then, the procedure $\textsc{PruneOrRepair}(G,\br,W, \Pi_{W \mapsto G}, \phi, \psi, R')$ given in \Cref{alg:betterwitness} either outputs
\begin{enumerate}
    \item  \label{prop:returnsCut} a set $S \subseteq V(G)$ with $R' \leq \vol_G(S) + \br(S) \leq 8R/\psi$ where $|E_G(S, \overline{S})| < {\phi} (\vol_G(S) + \br(S))$, or
    \item a new vector $\br' \in \mathbb{N}_{\geq 0}^{V(G)}$ and a new graph $W'$ with embedding $\Pi_{W' \mapsto G}$ that form an $(R', \phi, \psi')$-out-witness of $(G,\br')$ w.r.t. $\bgamma$, for $\psi' = \frac{\psi^2}{6}$.\label{prop:returnsWitness}
\end{enumerate}
The procedure can be implemented in time $\tilde{O}(R/(\psi^2 \phi))$.
\end{lemma}
\begin{remark}
Note that $\bgamma$ does not have to be passed as an argument to the procedure.
\end{remark}

\begin{algorithm}
\label{lne:defineCapacities} 
\DontPrintSemicolon
Define the flow problem $\mathcal{I} = (G, \bDelta, \bT, \bc)$ for $\bc =  \frac{16}{\psi\phi} \cdot \bOnes$; sink $\bT = \bDeg + \br$ and source function 
$\bDelta = \frac{8}{\psi}\br.$\label{lne:setUpFlowProblem}\\
Compute flow $\bflow$ by running Dinitz's Blocking Flow algorithm for $h = \frac{16\log(m)}{\psi\phi}$ rounds on instance $\mathcal{I}$.\label{lne:computeBlockingFlow}\\
\If(\label{lne:ifFlowIsRouted}){$\bflow$ is an $R'$-flow}{
    $W' \gets W$; $\Pi_{W' \mapsto G} \gets \Pi_{W \mapsto G}$; $\br' \gets \br$.\label{lne:initNewWit}\\
    Let $\mathcal{P}_{\bflow}$ be flow path decomposition of $\bflow$.\label{lne:specialFlowDecomposition}\\
    \ForEach(\label{mapUVPathsToWitnessEdges}){$u$-to-$v$ path $\pi \in \mathcal{P}_{\bflow}$}{
        Add edge $e=(u,v)$ to $W'$; $\Pi_{W' \mapsto G}(e) = \pi$; $\br'(u) \gets \br'(u) - 1$.\label{lne:updateWitness}
    }
    \Return $(W', \Pi_{W'\mapsto G}, \br')$
}\Else(\label{lne:algoEnterElse}){
    $S \gets \{v \in V(G) |\; \ex > 0 \}$.\label{lne:initS}\;
    \While{$|E_G(S, \overline{S})| \geq \phi(\vol_G(S) + \br(S))$}{
        $S \gets S \cup \{ v \in V \;|\; dist_{G_{\bflow}}(S, v) = 1\}$.\label{lne:augmentS}
    }
    \Return $S$\label{lne:returnSparseCut}
}\caption{$\textsc{PruneOrRepair}(G,\br, W, \Pi_{W \mapsto G}, \phi, \psi, R')$}
\label{alg:betterwitness}
\end{algorithm}

The algorithm works by setting up a flow instance $\mathcal{I}$ that tries to find for each unit $\br(v)$ a path from $v$ to an arbitrary other vertex in the graph while minimizing congestion and the number of flow paths ending in each vertex. We then run the Blocking Flow algorithm by Dinitz \cite{dinitz2006dinitz} for $h$ rounds on the flow instance $\mathcal{I}$. Our later analysis relies on the following well-known fact.

\begin{fact}\label{lma:propertiesOfF}
Given a flow instance $\mathcal{I}$ and height parameter $h$, the blocking flow algorithm by Dinitz run for $h$ rounds outputs a pre-flow $\bflow$ such that in the residual network $G_{\bflow}$ there is no path from any vertex $s \in V$ where $\ex(s) > 0$ to a vertex $t \in V$, with $\abs(t) < \bT(t)$ consisting of at most $h$ edges. 
\end{fact}

Note that we do not compute a $0$-flow which is achieved when Blocking Flow is run for $n$ rounds. Instead, we stop after only $h = O(\log(m)/(\psi\phi))$ rounds to ensure that the subprocedure can be implemented efficiently. Depending on whether the flow is then an $R'$-flow or not, we either use the flow to repair the witness graph $W$ by constructing $W'$ from $W$, or otherwise find a sparse cut $S$ in $G$.

We believe that the behaviour of the procedure is best understood by carefully inspecting the ensuing proof of \Cref{lma:asymPruning}. 

\paragraph{Proof of Case \ref{prop:returnsCut}.}
Let us first assume that the algorithm enters the else-statement starting in \Cref{lne:algoEnterElse}. Let us denote by $S_i$ the set $S$ constructed in the $i$-th iteration of the while-loop and by $S_0$ the set $S$ at initialization. Observe that we can alternatively characterize each $S_i$ by
$S_0 = \{v \in V(G) |\; \ex > 0 \}$ and for $i > 0$ by $S_i = \{ v \in V \;|\; dist_{G_{\bflow}}(S_0,v) \leq i\}$. Note that by \Cref{lma:propertiesOfF} and the definition of $\bT$, we have that $S_h$ absorbed at least $\vol_G(S_{h}) + \br(S_{h})$ units of flow. But the total amount of demand put at all vertices is $\| \bDelta \|_1 = \frac{8}{\psi}\| \br\|_1 \leq \frac{8}{\psi} R$ and so we must have $\vol_{G}(S_h) + \br(S_h) \leq \frac{8}{\psi} R$.

Let us assume first that the while-loop is terminated after $i \leq h$ iterations. Then, since $S_0 \subseteq S_i \subseteq S_h$, and since the vertices in $S_0$ are incident to at least $R'$ units of excess flow, we have $R' \leq \br(S_0) \leq \br(S_i)$ and combined with the while-loop condition, we clearly have that the cut $S_i$ returned in  \Cref{lne:returnSparseCut} is a valid output.

It remains to prove that the while-statement is indeed last entered for some $i \leq h$. We prove by contradiction by showing that $S_{i+1}$ has significantly larger volume than $S_i$ for each $i < h$ and therefore $S_h$ has volume larger than $8R/\psi$ which gives a contradiction by the argument above.

More precisely, we use that since, for $i \leq h$, $S_i$ is not a sparse cut, we have that $|E_G(S_i, \overline{S_i})|  \geq \phi (\vol_G(S_i) + \br(S_i))$. We next want to argue that the set $E_{G_{\bflow}}(S_i, \overline{S_i})$ is of comparable size to argue that $S_{i+1}$ is significantly larger in volume than $S_i$. But note that an edge $e$ in $E_G(S_i, \overline{S_i})$ does only not appear in $G_{\bflow}$ if $c = \bc(e)$ units of flow are routed in the edge. On the other hand, for any edge $\overleftarrow{e}$ in $E_G(\overline{S_i}, S_i)$, we have that an anti-parallel edge appears in $G_{\bflow}$ if 
any flow is routed on this edge. But note that the amount of flow leaving $S_i$ is clearly upper bound by $\bDelta(S_i)$. Thus,
\begin{align*}
    |E_{G_{\bflow}}(S_i, \overline{S_i})| &\geq \frac{c\cdot |E_G(S_i, \overline{S_i})| - \bDelta(S_i)}{c}\\
    &\geq \frac{c \cdot \phi (\vol_G(S_i) + \br(S_i)) - \frac{8}{\psi} \br(S_i)}{c} \geq \frac{1}{2} \phi (\vol_G(S_i) + \br(S_i)).
\end{align*}
for our choice of $c = \frac{16}{\psi\phi}$. We obtain by definition of $S_{i+1}$ that $\vol_G(S_{i+1}) \geq \left(1 + \frac{\phi}{2}\right)(\vol_G(S_{i}) + \br(S_{i})) \geq \left(1 + \frac{\phi}{2}\right)\vol_G(S_{i})$. Using induction, we thus get that 
\[
    \vol_G(S_{i}) + \br(S_{i}) \geq \left(1 + \frac{\phi}{2}\right)^i (\vol_G(S_{0}) + \br(S_{0})).
\]
Note that we can repeat this argument for all $i$, one can easily calculate that at level $h$ (where we use that $S_0$ is non-empty since otherwise we would have an $R'$-flow), we have $\vol_G(S_h) + \br(S_h) > m \geq 8R/\psi$. But this gives a contradiction, as desired.

\paragraph{Proof of Case \ref{prop:returnsWitness}.} We prove that $W'$, $\br'$ and $\Pi_{W' \mapsto G}$ form an $(R', \phi', \psi')$-out witness with respect to $\bgamma$. Let us therefore prove each property that is required by \Cref{def:alphaWitness} one-by-one:
\begin{enumerate}
    \item \uline{$\| \br' \|_1 \leq R'$:} We initialize $\br'$ to $\br$ in  \Cref{lne:initNewWit} and then decrease $\| \br' \|_1$ by $1$ in each iteration of \Cref{lne:updateWitness}. But since $\bflow$ is an $R'$-flow, the path decomposition of $\bflow$ holds at least $\|\br\|_1 - R'$ paths, each resulting in an iteration of the foreach-loop that executes \Cref{lne:updateWitness}.  
    
    \item \underline{$\forall v \in V(G)$, we have $\deg_{W'}(v) + \br'(v) \in [\deg_{G}(v),\frac{1}{\psi'}\deg_{G}(v)]$:} We first use that $\deg_W(v) + \br(v) \geq  \deg_{G}(v)$ by assumption on $W$. But note that we have $\br'(u)$ being equal to $\br(u)$ minus the number of edges added to $W'$ with tail in $u$, as can be seen from inspection of \Cref{lne:updateWitness}. Thus the lower bound holds.
    
    For the upper bound, we use that each vertex $v \in V$, has sink $\bT(v) = \deg_G(v) + \br(v)$. This upper bounds the number of paths that end in $v$ in the flow path decomposition and thus also edges added to $v$ with $v$ in its head. Thus, $\deg_{W'}(v) + \br'(v) \leq 2(\deg_{W}(v) + \br(v)) + \deg_{G}(v) \leq \frac{3}{\psi}\deg_{G}(v)$ by assumption on $W$. The last upper bound is significantly tighter than then the Lemma stipulates and we will use this tighter bound in proving the remaining properties.
    
    \item \uline{Expansion of cuts in $W'$:} Let us fix any cut $(S, \overline{S})$ where $\gamma(S) \leq \gamma(\overline{S})$. By assumption $|E_W(S, \overline{S})| + \br(S) \geq \psi (\vol_W(S) + \br(S))$. Let us do a case analysis:
    \begin{itemize}
        \item \uline{If $|E_W(S, \overline{S})| \geq \br(S)$:} Since $W' \supseteq W$, we have $|E_{W'}(S, \overline{S})| \geq |E_W(S, \overline{S})|$. At the same time, we have that $\vol_{W'}(S) + \br'(S) \leq \frac{3}{\psi} \vol_G(S)$ by the previously obtained degree bound. But from the guarantees on $W$ and $\br$, we thus have that $\vol_{W'}(S) + \br'(S) \leq \frac{3}{\psi} (\vol_{W}(S) + \br(S))$. Combining these insights, we obtain
        \[
        |E_{W'}(S, \overline{S})| \geq \frac{\psi}{2} (\vol_W(S) + \br(S)) \geq \frac{\psi^2}{6} (\vol_{W'}(S) + \br'(S)).
        \]
        \item \uline{If $|E_W(S, \overline{S})| < \br(S)$ \textbf{and} $\br'(S) > \frac{1}{2} \br(S)$:} We immediately get that
        \[
            |E_{W'}(S, \overline{S})| + \br'(S) > |E_{W}(S, \overline{S})| + \frac{1}{2}\br(S) \geq \frac{1}{2}\psi(\vol_W(S) + \br(S)) \geq \frac{\psi^2}{6}(\vol_{W'}(S) + \br'(S))
        \]
        where we use in the last inequality that $\vol_{W}(S) + \br(S) \geq \deg_G(v) \geq \frac{\psi}{3}(\vol_{W'}(S) + \br'(S))$ by assumption on $W$ and the last property.
        \item \underline{If $|E_W(S, \overline{S})| < \br(S)$ \textbf{and} $\br'(S) \leq \frac{1}{2} \br(S)$:} We have $\frac{8}{\psi}(\br(S) - \br'(S)) \geq \frac{4}{\psi} \br(S)$ paths in the flow decomposition (see \Cref{lne:specialFlowDecomposition}) with tails in $S$. But we also have that at most $\bT(S) = \vol_{G}(S) + \br(S)$ many of these edges have their head in $S$. The rest has their heads in $\overline{S}$. Thus $|E_{W'}(S, \overline{S})| \geq \frac{4}{\psi} \br(S) - \bT(S) \geq \frac{4}{\psi} \br(S) - (\vol_{G}(S) + \br(S))$. 
        
        But by assumption on $W$ and the current case assumption, we have $\vol_G(S) \leq \vol_W(S) + \br(S) \leq \frac{1}{\psi}(|E_W(S, \overline{S})| + \br(S)) < \frac{2}{\psi}\br(S)$. Thus, 
        $|E_{W'}(S, \overline{S})| \geq \frac{1}{\psi} \br(S) \geq \frac{1}{4}(\vol_G(S) + \br(S)) \geq \frac{\psi}{12}(\vol_{W'}(S) + \br'(S))$ (where we use the degree bound on $W'$ from the previous property in the last inequality).
        
    \end{itemize}
    
    \item \underline{$\Pi_{W' \mapsto G}$ has congestion at most $\frac{1}{\psi'\phi}$:} This follows straight-forwardly from the congestion of $\Pi_{W \mapsto G}$ and the fact that the embedding paths added to embed the new edges in $W'$ are taken from the flow path decomposition where the flow is routed through edges with capacities $\frac{16}{\psi \phi}$.
\end{enumerate}

\paragraph{Runtime Analysis.} Let us first analyze the run-time required to find the pre-flow $\bflow$. We assume for this section that the reader has basic familiarity with the classic Blocking Flow algorithm by Dinitz. This algorithm maintains a pre-flow $\bflow$ initialized to carry zero flow on every edge. Then, in each round a BFS algorithm is performed from an artificial super-source vertex $s$ on the residual graph $G'_{\bflow}$ obtained from the current $G_{\bflow}$ after adding the super-source vertex $s$ with an edge from $s$ to each vertex $v$ with residual capacity set equal to the current excess $\ex(v)$. Then, whenever the BFS discovers a new vertex $w$ with $\abs(w) < \bT(w)$, the algorithm can take a new flow path from the vertex $v$ after $s$ on the BFS tree path between $s$ and $w$ and add the flow path to $\bflow$ where the amount of flow is equal to the minimum residual capacity of any edge on the path. Any edge that has its residual capacity during this round decreased to $0$ remains removed from the graph that the BFS is performed on. 

Using this implementation, it is straight-forward to see that the BFS only explores out-edges in $G'_{\bflow}$ incident to $s$ and vertices where $\abs(w) = \bT(w)$. But the total volume of the latter set of vertices is at most $O(\vol_G(S_{h}))$ which we analyzed earlier to be at most $O(R/\psi)$. Since the number of edges incident to $s$ is at most $R$, we can conclude that each round consist of a BFS over $O(R/\psi)$ many edges along with the flow routing described above. Using a cut-link tree to route the flows, each round can thus be executed in time $\tilde{O}(R/\psi)$. The run-time for $h$ rounds of Blocking Flow is thus $\tilde{O}(R/(\psi^2\phi))$.

Finally, it is not hard to see that the if-condition in \Cref{lne:returnSparseCut} and the construction and updates of the set $S$ in Lines \ref{lne:initS} and \ref{lne:augmentS} can be done in $\vol_G(S)+\br(S)$ per iteration. But recall that $\vol_G(S)+\br(S) \leq O(R/\psi)$ and there are at most $h$ for-loop iteration.

\section{Maintaining Directed Expander Decomposition via Batching}
\label{sec:maintainBatching}

We now give the algorithm and analysis behind our main result in \Cref{thm:mainTheoremRand}. 

\paragraph{High-Level Algorithm.} The algorithm for \Cref{thm:mainTheoremRand} works by maintaining an expander decomposition $\mathcal{X} = \{X_1, X_2, \dots, X_{\tau}\}$ for graph $G$ at all times and for each expander $X_i$ it  batches updates to the graph $G[X_i]$ using standard batching techniques. This allows us to leverage the pruning algorithm from \Cref{lma:asymPruning} in the most effective way. 

More precisely, we maintain $L_{max} + 1$ levels of update batches in the algorithm for each set $X \in \mathcal{X}$. For each $X \in \mathcal{X}$, the algorithm maintains
\begin{itemize}
    \item a family of witness graphs $W_X = \{ W_{X,0}, W_{X,1}, \dots, W_{X, L_{max}}\}$.
    \item a family of vectors $R_X = \{\br_{X,0}, \br_{X,1}, \dots, \br_{X, L_{max}}\}$ where each vector lives in $\mathbb{N}_{\geq 0}^{V(G)}$ (but is supported only on $X$) and keeps track of the updates that need to be dealt with in each level.
\end{itemize}

\paragraph{Initialization.} To initialize, we set $\mathcal{X}$ to consist only of the set $V$, and set, for each $l$,  $W_{V,l}$ to the empty graph, $\br_{V,l} = \bDeg$, let $\psi_{L_{max}} = \frac{\psi_{CMG}}{2}$ and 
$\psi_l = \frac{\psi_{l+1}^{4}}{144}$ for $0 \leq l < L_{max}$. We initialize vector $\bgamma_V = \bDeg$ and the set $R$ to be the empty set. We then invoke procedure $\textsc{Update}(t)$ for $t=0$ which is described in the next paragraph.

\begin{algorithm}
\caption{Pseudocode to Process Updates}
\label{alg:updateAlgo}
\DontPrintSemicolon
\SetKwInOut{Input}{Input}
\SetKwInOut{Output}{Output}
\SetKw{Break}{Break}
\SetKw{Return}{Return}
\SetKwProg{Proc}{procedure}{}{}
\small
\Proc{$\textsc{ApplyUpdate}(u)$}{
    \lIf(\label{lne:initialIfCheck}){$u$ encodes the deletion $(x,y)$ where $x \in X, y \in Y$ for $X \neq Y \in \mathcal{X}$}{
    	\Return
    }
    \Else{
        Let $X \in \mathcal{X}$ be the cluster that the update $u$ is applied to.\\
        
       	\For{$l \in \{0, 1, \ldots, L_{max}\}$}
        {
            \If(\label{lne:uEncodesDel}){$u$ encodes the deletion of an edge $e = (x,y)$}{
    	        \ForEach(\label{lne:forEachLoopDeletion}){$(a,b) \in \Pi_{W_{X,l} \mapsto G}^{-1}(e)$}{
    	            Increment $\br_{X,l}(a)$ and $\br_{X,l}(b)$; Delete $(a,b)$ from $W_{X,l}$. 
    	        }
        	}
        	\lElseIf(\label{lne:uEncodesIns}){$u$ encodes self-loop insertion at vertex $x$}{
        	     Increase $\br_{X,l}(x)$ by $2$.\label{lne:addMassBcInsertion}
        	}\ElseIf(\label{lne:uEncodesSplit}){$u$ encodes a vertex split of $v$ with new vertex $v'$}{
        	    Let $W_{X,l}$ be the graph over the vertex set $X$ that now includes $v'$ and with the edge set obtained by remapping all $xy$-paths in $\Pi_{W_{X,l} \mapsto G}$ to $(x,y)$ edges in $W_{X,l}$.\\
        	    $\br_{X,l}(v) \gets \br_{X,l}(v) + \deg_{W_{X,l}}(v'); \br_{X,l}(v') \gets \deg_{W_{X,l}}(v'); \bgamma_{X}(v') \gets 0$.\label{lne:addALittleExtraSplit}\\
        	     \ForEach{$e=(x,y) \in E(W_{X,l})$ where $v' \in \Pi_{W_{X,l} \mapsto G}(e)$}{
        	        Increment $\br_{X,l}(x)$ and $\br_{X,l}(y)$; Delete $e$ from $W_{X,l}$.\label{lne:changeW}
        	    }
        	}
        	\While(\label{lne:whileLoopCappingCapacitiies}){there exists $v \in V$ with $\deg_{W_{X,l}}(v) + \br_{X,l}(v) > \frac{ \deg_{G[X]}(v)}{\psi_l}$}{
        	    \ForEach{edge $(u,v)$ or $(v,u)$ in $E(W_{X,l})$}{
        	        $\br_{X,l}(u) \gets \br_{X,l}(u) + 1$; Remove the edge from $W_{X,l}$.\label{lne:changeW2}
        	    }
        	    $\br_{X,l}(v) \gets \deg_{G[X]}(v)/\psi_l$.
        	}
        	
        }
    }
}

\Proc{$\textsc{Update}(t)$}{
    \lIf{$t > 0$}{Invoke $\textsc{ApplyUpdate}$ to the $t$-th update.\label{lne:firstApplyUpdate}}
    \While(\label{lne:mainWhileLoop}){$\exists X \in \mathcal{X}$ where $\exists l \in [L_{max}]$ with $\|\br_{X,l}\|_1 \geq \frac{\psi_{l}}{8} |E(G[X])|^{l/L_{max}}$}{
        Let $X$ and $l$ be such that the while condition holds for them and $l$ is the maximum integer for which the condition holds.\label{lne:chooseXandL}\\
        \If{$l = L_{max}$}{
            Run $\textsc{CutOrEmbed}(G[X], \phi,  \frac{\psi_{L_{max}}}{16} |E(G[X]|)$ and if it returns a cut, let this cut be denoted by $S$; otherwise, set $W_{X,l'}, \br_{X,l'}$ equal to the returned witness and $\br$ vector as specified in \Cref{thm:cutMatching} for each $l' \in \{0, 1, \ldots, L_{max}\}$ and $\bgamma_X= \mathbf{deg_{G[X]}}$.\label{lne:newWitness}
        }\Else{
            Run procedures $\textsc{PruneOrRepair}(G[X],\br_{X, l +1} , W_{X, l +1}, \Pi_{W_{X, l +1} \mapsto G[X]}, \phi, \psi_{l+1}, \frac{\psi_{l}}{32}|E(G[X])|^{l/L_{max}})$, 
            $\textsc{PruneOrRepair}(\overleftarrow{G}[X],\br_{X, l +1}, W_{X, l +1}, \Pi_{W_{X, l +1} \mapsto \overleftarrow{G}[X]}, \phi, \psi_{l+1}, \frac{\psi_{l}}{32}|E(G[X])|^{l/L_{max}})$; if either of them returns a cut, let that cut be stored in $S$; otherwise let the witnesses and $\br$ vectors returned be denoted by $(W_1, \br_1)$ and $(\overleftarrow{W_2}, \br_2)$; set $W_{X,l}$ to $W_1 \cup W_2$ and $\br_{X,l}$ to $\br_1 + \br_2$.\label{lne:computeWitnessViaHigherLvlExpander}\label{lne:newWitness2}
            
        }
        
        \If(\label{lne:elseSplitWitness}){a cut $S$ was returned}{
                Add the smaller set of edges $E_{G}(X \setminus S, S)$ or $E_{G}(S, X \setminus S)$ to $R$.\label{lne:addSmallerSideToR}\\ \lForEach(\label{lne:forEachRemoveSparseCut}){edge $e \in E_{G}(X \setminus S, S) \cup  E_{G}(S, X \setminus S)$}{
                Delete $e$ via $\textsc{ApplyUpdate}$.\label{lne:artificialUpdates}
            }
            Replace $X$ in $\mathcal{X}$ by $S$ and $X \setminus S$.\label{lne:updateXinExpDecomp}\\
            \ForEach{$l \in \{0, 1, \ldots, L_{max}\}$ and $X' \in \{ S, X \setminus S\}$}
            {
                Let $W_{X',l}$ be assigned the induced graph $W_{X,l}[X']$; let  $\br_{X', l}$ be the vector $\br_{X,l}$ restricted to the set $X'$; let $\bgamma_{X',l}$ be the vector $\bgamma_{X}$ restricted to $X'$.\label{lne:induceOnSubcluster}\\
            }
        }
    
    }
}
\end{algorithm}

\paragraph{Update.}

The update algorithm given in \Cref{alg:updateAlgo} consists of an utility procedure $\textsc{ApplyUpdate}$ and the main procedure $\textsc{Update}$. The procedure $\textsc{ApplyUpdate}$ handles intermediate updates to the low-level data structures during the processing of an update to the graph $G$. The procedure $\textsc{Update}$ computes the new expander decomposition after executing an update to the graph. Again, we believe that the procedures are best understood by analyzing them.

\paragraph{Analysis ($\textsc{ApplyUpdate}$).} We start by arguing about the procedure $\textsc{ApplyUpdate}$ which processes updates to $G$ and forwards them to the witness graphs. 

\begin{claim}\label{clm:applyUpdateWorksWell}
For any invocation of procedure $\textsc{ApplyUpdate}$, for any cluster $X \in \mathcal{X}$, and level $l \in [L_{max}]$, if
$\min\{|E_{W_{X,l}}(S, \overline{S}) + \br_{X,l}(S)|, |E_{W_{X,l}}(\overline{S}, S) + \br_{X,l}(S)|\} \geq \psi_l (\vol_{W_{X,l}}(S) + \br_{X,l}(S))$
whenever $\bgamma_X(S) \leq \bgamma_X(\overline{S})$ where $\overline{S}= X \setminus S$ holds before the invocation, then it also holds after the invocation for the updates in $W_{X,l}$ and $\br_{X,l}$.
\end{claim}
\begin{proof}
First, we observe that if the update $u$ satisfies the condition of the if-statement in \Cref{lne:initialIfCheck}, then no changes are executed and we can therefore ignore the case.

Otherwise, the update $u$ affects a cluster $X \in \mathcal{X}$. We use superscripts $OLD$ and $NEW$ to denote variables in the state just before the invocation and just after the invocation of $\textsc{ApplyUpdate}$ respectively. Let us consider any cut $(S, \overline{S})$ where $\bgamma(S) \leq \bgamma(\overline{S})$. We define $S^{iso} \subseteq S$ to be the vertices in $S$ that are isolated in $W^{NEW}_{X,l}$ after $\textsc{ApplyUpdate}$. 

Let us first analyze the case when $S^{iso} = \emptyset$. In this case, we have that no vertex $v \in S$, has entered the while-loop starting in \Cref{lne:whileLoopCappingCapacitiies}. Let us do a case analysis for the udpate types:
\begin{itemize}
    \item \uline{For $u$ edge deletion:} Since no $v \in S$ entered the while-loop in \Cref{lne:whileLoopCappingCapacitiies}, we have in this case that $|E_{W^{OLD}_{X,l}}(S, \overline{S})| + \br^{OLD}_{X,l}(S) \leq |E_{W^{NEW}_{X,l}}(S, \overline{S} )| + \br^{NEW}_{X,l}(S)$ since for each deleted edge from $E_{W^{OLD}_{X,l}}(S, \overline{S})$, the procedure places one unit to $\br_{X,l}$ on an entry in $S$ (and one on an entry in $\overline{S}$). Further, we have $\vol_{W_{X,l}^{NEW}}(S) + \br_{X,l}^{NEW}(S) = \vol_{W_{X,l}^{OLD}}(S) + \br_{X,l}^{OLD}(S)$ since we add to the vector $\br_{X,l}$ what is lost in volume and do not enter the final while-loop by assumption. Thus
    \begin{align*}
    |E_{W^{NEW}_{X,l}}(S, \overline{S} )| + \br^{NEW}_{X,l}(S) &\geq
     |E_{W^{OLD}_{X,l}}(S, \overline{S} )| + \br^{OLD}_{X,l}(S) \\ &\geq \psi_l (\vol_{W_{X,l}^{OLD}}(S) + \br_{X,l}^{OLD}(S)) 
    \geq \psi_l (\vol_{W_{X,l}^{NEW}}(S) + \br_{X,l}^{NEW}(S)).
    \end{align*}
    \item \uline{For $u$ self-loop insertion:} This case can be verified straight-forwardly.
    \item \uline{For $u$ a vertex split:} 
    Let vertex $v$ be split into $v$ and $v'$. Recall that we assume that no vertex $s \in S$ is isolated in $W^{NEW}_{X,l}$. Thus, since we delete all edges that have $v'$ on their embedding path to obtain $W_{X,l}^{NEW}$, we must have that $v' \not\in S$.
    
    The remaining case analysis can be made closely to the argument for $u$ being an edge deletion when paying special attention to the case where $v$ is in $S$ and one has to use that we add $\vol_{W_{X,l}}(v')$ to $\br_{X,l}(v)$.
\end{itemize}

To prove for the case where $S^{iso} \neq \emptyset$, note that we can use the proof above to show that the claim holds for the set $S \setminus S^{iso}$. It then remains to observe that adding a set of isolated vertices to any set $S'$ that has the properties of our claim, does not invalidate the claim as it only adds mass to the $\br_{X,l}$ vector. The claim in its full generality follows.
\end{proof}

For the rest of the analysis, we often look at the graph maintained internally by our data structure which is defined below.

\begin{definition}[Maintained Graph]
At any point in the algorithm, we let $G_U$ denote the graph $G$ after applying all the updates to $G$ on which the procedure $\textsc{ApplyUpdate}$ was run (also the ones issued by the algorithm in \Cref{lne:artificialUpdates}).
\end{definition}
\begin{remark}
Technically, the definition of $G_U$ is not well-defined for the times spent within the procedure $\textsc{ApplyUpdate}$ but we avoid such ambiguities by only using $G_U$ when talking about times before or after such procedure calls.
\end{remark}

Observe that by the definiton above, we have at the end of each stage, i.e. after processing each the current update to $G$, that $G_U \subseteq G$ since we invoke $\textsc{ApplyUpdate}$ on each update to $G$ within the same stage (see \Cref{lne:firstApplyUpdate}). We start by proving the following rather simple structural claim.

\begin{claim}\label{clm:embeddingIsNicelyMaintained}
Before and after any invocation of $\textsc{ApplyUpdate}$, we have that for every $X \in \mathcal{X}$ and level $l \in [ L_{max}]$, the embedding $\Pi_{W_{X,l} \mapsto G}$ maps each edge $(u,v)$ in $W_{X,l}$ to a $u$-to-$v$ path in $G_U[X]$.
\end{claim}
\begin{proof}
We note that by our initialization procedure, before the first invocation of $\textsc{Update}(t)$ (i.e. when $t = 0$), the claim holds. Next, we note that during each invocation of $\textsc{ApplyUpdate}$, if $u$ encodes an edge deletion, we remove all paths from $W_{X,l}$ that are embed into the affected edge (see the if-case in  \Cref{lne:uEncodesDel}). If $u$ encodes a vertex split of $v$ splitting of $v'$, then each embedding path that went through $v$ by having an edge $(x,v)$ entering and an edge $(v,y)$ leaving might no longer be a real path if exactly one of the endpoints is mapped to $v'$ instead of $v$. But in this case $v'$ is on the embedding path, and it is exactly such embedding paths that are removed in the if-case in \Cref{lne:uEncodesSplit}. Finally, it is easy to see that whenever we compute an entirely new witness and witness embedding (see \Cref{lne:newWitness2}), the embeddings are found in the current graph $G_U[X] = G[X]$.
\end{proof}

\paragraph{Analysis (Correctness).} Before we can argue about correctness, let us make the following definitions.

\begin{definition}[Subcluster]
Given a vertex $x$ in $G$ at any stage $t$, we say that it \emph{originates} from a vertex $y$ at an earlier stage $t' \leq t$ in $G$ if $x$ was obtained from a sequence of adversarial vertex splits applied to $y$. Given a cluster $X \in \mathcal{X}$ at any stage $t$ and a cluster $Y \in \mathcal{X}$ at a later stage $t' \geq t$, we say $Y$ is a \emph{subcluster} of $X$ if all vertices in $Y$ originate from vertices in $X$.
\end{definition}

\begin{definition}\label{def:xInit}
For any cluster $X \in \mathcal{X}$ and level $l$, let $X^{INIT, l}$ be the most recent subcluster of $X$ such that \Cref{lne:newWitness} or \Cref{lne:newWitness2} was executed on $X^{INIT, l}$ and $l$ and the witness $W_{X^{INIT, l}, l}$ and vector $\br_{X^{INIT, l}, l}$ were (re-)initialized during the execution of this line.
\end{definition}

We can now argue that \Cref{alg:updateAlgo} correctly maintains witness graphs. 

\begin{invariant}\label{inv:mainInv}
Every time the condition of the while-loop starting in \Cref{lne:mainWhileLoop} is evaluated, we have for every $X \in \mathcal{X}$ and level $l \in [L_{max}]$, that
$W_{X,l}$ is a $(\infty, \phi, \psi_l)$-witness of $(G[X], \br_{X,l})$ with respect to $\bgamma_X$. Whenever $W_{X,l}$ is (re-)initialized, we further have that it is a $(\infty, \phi, 2\psi_l)$-witness of $(G[X], \br_{X,l})$ with respect to $\bgamma_X$.
\end{invariant}
\begin{proof}
We prove the invariant by induction over the times that \Cref{lne:mainWhileLoop} is evaluated. 

\uline{Base case:} Before the first time that the while-loop condition is evaluated, we have by our initialization procedure that $\mathcal{X} = \{V\}$ and that for each $l$, vector $\br_{V,l} = \bDeg$ (also since \Cref{lne:firstApplyUpdate} is skipped when we invoke $\textsc{Update}(0)$). Thus, we trivially have that $W_{V, l} = (V, \emptyset)$ is a $(\infty, \phi, \psi_l)$-witness of $(G, \br_{V,l})$ w.r.t. $\bgamma_V = \bDeg$, which establishes the base case.

\uline{Inductive Step: } For any cluster $X \in \mathcal{X}$ and $l$, let $X^{INIT,l}$,  $W_{X^{INIT,l}, l}$, and $\br_{X^{INIT,l}, l}$ be defined as in \Cref{def:xInit}. 

Consider first the case that $W_{X^{INIT,l}, l}$ was (re-)initialized after the last time that the invariant held when the while-loop condition was executed. Then, in between these two times, a single iteration of the while-loop in \Cref{lne:mainWhileLoop} is performed on exactly $X = X^{INIT,l}$. We distinguish by cases:
\begin{itemize}
    \item \underline{If $W_{X,l}$ was (re-)initialized in \Cref{lne:newWitness}:} Then by \Cref{thm:cutMatching}, we have that $W_{X, l}$ is a $(\infty, \phi, 2\psi_{L_{max}})$-witness of $(G[X], \br_{X,l})$ with respect to $\mathbf{deg_{G[X]}}$ where $\psi_{L_{max}} \geq \psi_{l}$. Since at the same time, the algorithm (re-)sets $\bgamma_X = \mathbf{deg_{G[X]}}$, the invariant follows.
    
    \item \uline{Otherwise:} we have that $W_{X,l}$ was (re-)initialized in \Cref{lne:newWitness2}. But this implies that $W_{X, l+1}$ was not (re-)initialized since the last time that the while-loop condition was executed; and clearly also $G$ and $\bgamma_X$ were not changed since then. Using further the maximality of $l$ (see \Cref{lne:chooseXandL}), thus, we can use the induction hypothesis to argue that $W_{X, l+1}$ is a $(R, \phi, \psi_{l+1})$-witness of $(G, \br_{X, l+1})$ w.r.t. $\bgamma_X$ for $R \leq \frac{\psi_{l+1}}{8} |E(G[X])|^{(l+1)/L_{max}} \leq \frac{\psi_{l+1}}{8} |E(G[X])|$. 
    
    Thus, the assumptions of \Cref{lma:asymPruning} are satisfied when the algorithm invokes the two procedures executed in \Cref{lne:newWitness2} to obtain $W_{X,l}$ and $\br_{X,l}$, return witnesses $W_1$ and $\overleftarrow{W_2}$ along with vectors $\br_1$ and $\br_2$. By \Cref{lma:asymPruning}, $W_1$ (analogously $\overleftarrow{W_2}$) is a $(\infty, \phi, \frac{\psi_{l+1}^2}{6})$-out-witness of $(G[X], \br_1)$ w.r.t. $\bgamma_X$ (analogously $(\overleftarrow{G[X]}, \br_2)$). 
    
    It remains to verify that $W_{X,l} = W_1 \cup W_2$ is a $(\infty, \phi, \psi_{l}/2)$-witness of $(G, \br_1 + \br_2)$ w.r.t. $\bgamma_X$. We note that the witness properties given in \Cref{def:alphaWitness} are trivial to prove except for Property \ref{prop:alphaWitness2} which we next prove carefully.
    
    For convenience, we define $\widehat{\psi} = \frac{\psi_{l+1}^2}{6}$. Consider first any cut $(S, \overline{S})$ where  $\bgamma_X(S) \leq \bgamma_X(\overline{S})$ (the vector we use in \Cref{lma:asymPruning}). By properties of $W_1$, we have that $|E_{W_1}(S, \overline{S})| + \br_1(S) \geq \widehat{\psi}(\vol_{W_1}(S) + \br_1(S))$. But note that by the properties of $W_1$ and $W_2$, we have 
    \[
        \vol_{W_{X,l}}(S) + \br_{X,l}(S) \leq (\vol_{W_1}(S) + \br_1(S)) (1+ \frac{1}{\widehat{\psi}}) \leq \frac{2}{\widehat{\psi}} \cdot (\vol_{W_1}(S) + \br_1(S)). 
    \]
    Combining these insights, we can conclude that 
    \[|E_{W_{X,l}}(S, \overline{S})| + \br_{X,l}(S) \geq \frac{\widehat{\psi}^{2}}{2} \left(\vol_{W_{X,l}}(S) + \br_{X,l}(S)\right).\]
    Using same analysis on $W_2$ establishes that 
    \[|E_{W_{X,l}}(\overline{S}, S)| + \br_{X,l}(S) \geq \frac{\widehat{\psi}^{2}}{2} \left(\vol_{W_{X,l}}(S) + \br_{X,l}(S)\right).\]
    
    Using that $\frac{\widehat{\psi}^{2}}{2} = \frac{\psi_{l+1}^4}{72} = 2\psi_l$, we can therefore conclude that $W_{X,l}$ is a $(\infty, \phi, \psi_l)$-witness of $(G[X], \br_{X,l})$.
\end{itemize}

It remains to argue for the invariant in the case where $W_{X^{INIT,l}, l}$ was not (re-)initialized after the last time that the invariant held when the while-loop condition was executed). 

We consider the following cases:
\begin{itemize}
    \item \uline{If a new stage has started, after the last time that the invariant held:} in this case an adversarial update $u$ was applied to $G$. We note that $\textsc{ApplyUpdate}$ preserves the cut-expansion properties by \Cref{clm:applyUpdateWorksWell}, and the embedding property follows from \Cref{clm:embeddingIsNicelyMaintained}. Further, it is not hard to see that the quantity $\deg_{G[X]}(v) + \br_{X,l}(v)$ does not decrease due to invoking procedure $\textsc{ApplyUpdate}$ except if the quantity exceeds the degree of $v$ in $G[X]$ by a large quantity in which case it is normalized (in the while-loop starting in \Cref{lne:whileLoopCappingCapacitiies}) which provides us with the degree preserving property of witness $W_{X,l}$.
    
    \item \uline{If no new stage has started:} then the underlying graph $G[X]$ was not changed. The only possible change to the cluster $X$ is that it might have been undergoing changes due to the updates applied in \Cref{lne:forEachRemoveSparseCut} and/or might have been induced. But note that we argued above that applying updates via $\textsc{ApplyUpdate}$ does not affect correctness, and it is not hard to verify that inducing does not affect correctness either since we induce in such a way that already now edge crosses between the newly induced clusters.
\end{itemize}
\end{proof}

\begin{corollary}
\label{lma:witnessCorrectlyMaintained}
At the end of every stage $t$, for any $X \in \mathcal{X}$ and level $l \in \{0,1,\ldots, L_{max}\}$, 
$W_{X,l}$ is a $(\frac{\psi_{l}}{8}|E(G[X])|^{l/L_{max}} - 1, \phi, \psi_l)$-witness of $(G[X], \br_{X,l})$ w.r.t. $\bgamma_X$.
\end{corollary}
\begin{proof}
Assuming that the algorithm finishes in finite time, we have that after each while-loop the claim holds by the while-loop condition and \Cref{inv:mainInv}.
\end{proof}

Overall correctness follows by \Cref{lma:witnessCorrectlyMaintained} for all $X \in \mathcal{X}$ and level $0$ combined with \Cref{clm:conductanceForRNull}.

\paragraph{Analysis (Set $R$).} From the algorithm, it is clear that $R$ is a set that only grows over time since the only place in the algorithm where edges are added to $R$ is in \Cref{lne:addSmallerSideToR}. We further note that whenever we add edges to $R$ before we decompose $X$ into $S$ and $X \setminus S$, by \Cref{lma:asymPruning}, we add a batch of at most $\phi \min\{ \vol_{G[X]}(S) + \br_{X,l}(S), \vol_{G[X]}(X \setminus S) + \br_{X,l}(X \setminus S)\} \leq \frac{2\phi}{\psi_{l}} \min\{ \vol_{G[X]}(S), \vol_{G[X]}(X \setminus S)\}$ edges where the inequality follows from \Cref{lma:witnessCorrectlyMaintained}. Thus, we can charge the cut to the edges on the smaller side. Since each edge appears at most $O(\log m)$ times on the smaller side of the cut, we can bound the total size of $R$ by $\tilde{O}\left(\frac{\phi}{\psi_{CMG}\psi_{0}} m\right)$.

\paragraph{Analysis (Run-time).} Finally, let us argue about the total run-time of the algorithm.

\begin{claim}\label{clm:increaseInVectors}
The total amount that the vectors $\br_{X,l}$ (over all $X$ and $l$) are increased in the procedure $\textsc{ApplyUpdate}$ is $\tilde{O}(m \cdot L^2_{max} \cdot \frac{1}{\phi \psi_{0}})$.
\end{claim}
\begin{proof}
We distinguish by updates. For edge deletions, we increase the vectors $\br_{X,l}$ by $2$ for each edge in $W_{X,l}$ embed into the edge deleted. Since we maintain $W_{X,l}$ to be a witness by \Cref{lma:witnessCorrectlyMaintained}, we conclude that there are at most $\frac{1}{\phi \psi_{l}} \leq \frac{1}{\phi \psi_{0}}$ such edges, and therefore the total contribution by all of the at most $m$ edge deletions is $O(m \cdot L_{max} \cdot \frac{1}{\phi \psi_{0}})$. Self-loop insertions increase vectors on each level by $2$ and therefore we have total increase $O(m \cdot L_{max})$ from self-loop insertions.

For vertex splits where $v$ is split into $v$ and $v'$, we add $2\vol_{W_{X,l}}(v') = O(\frac{1}{\psi_0} \vol_G(v'))$ directly to the vector entries of $v$ and $v'$ by  \Cref{lma:witnessCorrectlyMaintained}. Additionally, we remove all embedding paths through the vertex $v'$. But note that the number of such embedding paths by  \Cref{lma:witnessCorrectlyMaintained} can be at most $O(\vol_{G}(v') \cdot L_{max} \cdot \frac{1}{\phi \psi_{0}})$. But since each edge can be on the side of the vertex split with smaller volume, i.e. incident to $v'$, for at most $O(\log(m))$ times, we have that the total increase from vertex splits is bound by $\tilde{O}(m \cdot L_{max} \cdot \frac{1}{\phi \psi_{0}})$.

Finally, we account for increases in $\br_{X,l}$ vectors due to the while-loop starting in \Cref{lne:whileLoopCappingCapacitiies}. We start by observing that whenever a vertex $v$ is isolated in the while-loop in \Cref{lne:whileLoopCappingCapacitiies}, the amount that we increase the vector $\br_{X,l}$ (for $v \in X$) is upper bound by the current degree $\deg_{W_{X,l}}(v)$. By induction on the invocations of $\textsc{ApplyUpdate}$, we can bound $\deg_{W_{X,l}}(v)$ by $O(\deg_{G[X]}(v)/\psi_l)$.

But note that since we prove that immediately after the re-initialization of each $W_{X,l}$, we have that it is a $(\infty, \phi, 2\psi_l)$ witness of $(G, \br_{X,l})$ (see \cref{inv:mainInv}), we have that $\deg_{W_{X,l}}(v) + \br_{X,l}(v) \leq \deg_{G[X]}(v)/(2\psi_l)$. But since a vertex $v$ only gets isolated in \Cref{lne:whileLoopCappingCapacitiies} if $\deg_{W_{X,l}}(v) + \br_{X,l}(v) > \deg_{G[X]}(v)/\psi_l$, then either $\deg_{W_{X,l}}(v) + \br_{X,l}(v)$ has increased by a factor of at least $4/3$ or $\deg_{G[X]}(v)/\psi_l$ has decreased by factor at least $2/3$.

Let us first argue about the quantity $\deg_{W_{X,l}} + \br_{X,l}$. It is not hard to see that when edges are deleted from $W_{X,l}$ (either in \Cref{lne:changeW} or in \Cref{lne:changeW2}), the algorithm compensates by adding an additional unit to $\br_{X,l}$ at the endpoints of the deleted edge. Thus, $\deg_{W_{X,l}} + \br_{X,l}$ remains unchanged. However, the quantity $\deg_{W_{X,l}} + \br_{X,l}$ might be changed in \Cref{lne:addMassBcInsertion} or \Cref{lne:addALittleExtraSplit}. Both times, the quantity increases, in the former by $2$ in the coordinate of the vertex where a new self-loop is added, and in the latter by the degree of the vertex (in $W_{X,l}$ which is at most $O(1/\psi_l)$ times the degree of the same vertex in $G$) that is split off. We can thus bound the total amount of increases in $\|\deg_{W_{X,l}} + \br_{X,l}\|_1$ over all $X$ and $l$ by $\tilde{O}(L_{max} \cdot m/\psi_l)$ since each edge appears at most $O(\log(m))$ times on the smaller side of a vertex split. By our previous reasoning, this implies that these changes in $\deg_{W_{X,l}} + \br_{X,l}$ can increase the vector $\br_{X,l}$ over all levels $l$ and clusters $X$ by at most $\tilde{O}(L^2_{max} \cdot m /\psi^{2}_{0})$ (here we use that $\phi \leq \psi_0$ by assumption).

For the total number of changes to $\|\deg_{G[X]}\|_1$ over all $X$ and $l$, we can further straight-forwardly obtain the upper bound $\tilde{O}(m)$. Using the reasoning from before, we thus obtain a total of at most $\tilde{O}(L^2_{max} \cdot m /\psi_{0})$ in increase in vectors $\br_{X,l}$. 
\end{proof}

\begin{lemma}
The algorithm takes total time $\tilde{O}(m^{1+1/L_{max}} \cdot L_{max} \cdot (\log(m))^{4^{O(L_{max})}} /\phi^2)$.
\end{lemma}
\begin{proof}
Whenever the procedure $\textsc{PruneOrRepair}$ is run on a set $X \in \mathcal{X}$ and level $l \in [L_{max}]$, it does so since $\|\br_{X,l}\|_1 \geq \frac{\psi_{l}}{8} |E(G[X])|^{l/L_{max}}$ by the condition of the while-loop in \Cref{lne:mainWhileLoop}. It then re-sets $(W_{X,l}, \br_{X,l})$ in \Cref{lne:newWitness2} such that $\|\br_{X,l}\|_1 = \|\br_{1}\|_1 + \|\br_{2}\|_1 \leq 2 \cdot \frac{\psi_{l}}{32} |E(G[X])|^{l/L_{max}} = \frac{\psi_{l}}{16} |E(G[X])|^{l/L_{max}}$. Thus, each such computation decreases the $\ell_1$-sum of all vectors $\br_{X', l'}$ over all $X'$ and $l'$ by at least $\frac{\psi_{l}}{16} |E(G[X])|^{l/L_{max}}$.

But note that the invocation of $\textsc{PruneOrRepair}$ takes time $\tilde{O}(\|\br_{X,l+1}\|_1 / \psi_{0}^2\phi)$ by \Cref{lma:asymPruning}. Since we always pick the largest $l$ for which the while-loop condition in \Cref{lne:mainWhileLoop} is satisfied first, this implies that the run-time is at most $\tilde{O}(|E(G[X])|^{(l+1)/L_{max}} / \psi_{0}\phi)$. Thus, we can charge time spent in these invocations of $\tilde{O}(|E(G[X])|^{1/L_{max}}/\psi_{0}^2\phi) = \tilde{O}(m^{1/L_{max}}/\psi_{0}^2\phi)$ to each unit that we remove from $\br_{X,l}$ due to this invocation.

Combining this insight with the fact that initially $\|\br_{V,l}\|_1 = 2m$ for all levels $l$ and with the increase bound from \Cref{clm:increaseInVectors}, we can bound the total time spend for all such invocations by $\tilde{O}(m^{1+1/L_{max}} \cdot L_{max}^2 \cdot \frac{1}{\phi^2 \psi_{0}^3})$.

It remains to observe that by the analysis from \Cref{clm:increaseInVectors}, we can also bound the total run-time of all invocations of $\textsc{ApplyUpdate}$ by  $\tilde{O}(m \cdot L_{max} \cdot \frac{1}{\phi \psi_{0}})$. The time of all other operations is subsumed by the time spend on the invocations to either $\textsc{PruneOrRepair}$ or $\textsc{ApplyUpdate}$.
\end{proof}

\section*{Acknowledgements}

The authors are grateful for insightful discussions with Li Chen, Yang P. Liu, and Sushant Sachdeva that ultimately lead to the key insight in this article. We are very grateful for feedback from Richard Peng, Thatchaphol Saranurak, Jason Li and Simon Meierhans on an early draft of this article that helped to streamline presentation.

\bibliographystyle{alpha}
\bibliography{main}

\newcommand{\etalchar}[1]{$^{#1}$}
\begin{thebibliography}{vdBLN{\etalchar{+}}20}

\bibitem[BBG{\etalchar{+}}22]{bernstein2020fully}
Aaron Bernstein, Jan van~den Brand, Maximilian~Probst Gutenberg, Danupon
  Nanongkai, Thatchaphol Saranurak, Aaron Sidford, and He~Sun.
\newblock Fully-dynamic graph sparsifiers against an adaptive adversary.
\newblock {\em ICALP 2022}, 2022.

\bibitem[BGS20]{bernstein2020deterministic}
Aaron Bernstein, Maximilian~Probst Gutenberg, and Thatchaphol Saranurak.
\newblock Deterministic decremental reachability, scc, and shortest paths via
  directed expanders and congestion balancing.
\newblock In {\em 2020 IEEE 61st Annual Symposium on Foundations of Computer
  Science (FOCS)}, pages 1123--1134. IEEE, 2020.

\bibitem[BGS22]{bernstein2021deterministic}
Aaron Bernstein, Maximilian~Probst Gutenberg, and Thatchaphol Saranurak.
\newblock Deterministic decremental sssp and approximate min-cost flow in
  almost-linear time.
\newblock pages 1000--1008, 2022.

\bibitem[BGWN21]{bernstein2021decremental}
Aaron Bernstein, Maximilian~Probst Gutenberg, and Christian Wulff-Nilsen.
\newblock Decremental strongly connected components and single-source
  reachability in near-linear time.
\newblock {\em SIAM Journal on Computing}, (0):STOC19--128, 2021.

\bibitem[CDK{\etalchar{+}}21]{chalermsook2021vertex}
Parinya Chalermsook, Syamantak Das, Yunbum Kook, Bundit Laekhanukit, Yang~P
  Liu, Richard Peng, Mark Sellke, and Daniel Vaz.
\newblock Vertex sparsification for edge connectivity.
\newblock In {\em Proceedings of the 2021 ACM-SIAM Symposium on Discrete
  Algorithms (SODA)}, pages 1206--1225. SIAM, 2021.

\bibitem[CGL{\etalchar{+}}20]{chuzhoy2020deterministic}
Julia Chuzhoy, Yu~Gao, Jason Li, Danupon Nanongkai, Richard Peng, and
  Thatchaphol Saranurak.
\newblock A deterministic algorithm for balanced cut with applications to
  dynamic connectivity, flows, and beyond.
\newblock In {\em 2020 IEEE 61st Annual Symposium on Foundations of Computer
  Science (FOCS)}, pages 1158--1167. IEEE, 2020.

\bibitem[CGP{\etalchar{+}}20]{chu2020graph}
Timothy Chu, Yu~Gao, Richard Peng, Sushant Sachdeva, Saurabh Sawlani, and
  Junxing Wang.
\newblock Graph sparsification, spectral sketches, and faster resistance
  computation via short cycle decompositions.
\newblock {\em SIAM Journal on Computing}, (0):FOCS18--85, 2020.

\bibitem[CHI{\etalchar{+}}16]{chechik2016decremental}
Shiri Chechik, Thomas~Dueholm Hansen, Giuseppe~F Italiano, Jakub
  {\L}{\k{a}}cki, and Nikos Parotsidis.
\newblock Decremental single-source reachability and strongly connected
  components in o (m√ n) total update time.
\newblock In {\em 2016 IEEE 57th Annual Symposium on Foundations of Computer
  Science (FOCS)}, pages 315--324. IEEE, 2016.

\bibitem[Chu21]{chuzhoy2021decremental}
Julia Chuzhoy.
\newblock Decremental all-pairs shortest paths in deterministic near-linear
  time.
\newblock In {\em Proceedings of the 53rd Annual ACM SIGACT Symposium on Theory
  of Computing}, pages 626--639, 2021.

\bibitem[CK19]{chuzhoy2019new}
Julia Chuzhoy and Sanjeev Khanna.
\newblock A new algorithm for decremental single-source shortest paths with
  applications to vertex-capacitated flow and cut problems.
\newblock In {\em Proceedings of the 51st Annual ACM SIGACT Symposium on Theory
  of Computing}, pages 389--400, 2019.

\bibitem[CKK{\etalchar{+}}06]{chawla2006hardness}
Shuchi Chawla, Robert Krauthgamer, Ravi Kumar, Yuval Rabani, and D~Sivakumar.
\newblock On the hardness of approximating multicut and sparsest-cut.
\newblock {\em computational complexity}, 15(2):94--114, 2006.

\bibitem[CKL{\etalchar{+}}22]{chen2022maximum}
Li~Chen, Rasmus Kyng, Yang~P Liu, Richard Peng, Maximilian~Probst Gutenberg,
  and Sushant Sachdeva.
\newblock Maximum flow and minimum-cost flow in almost-linear time.
\newblock {\em arXiv preprint arXiv:2203.00671}, 2022.

\bibitem[CS21]{chuzhoy2021deterministic}
Julia Chuzhoy and Thatchaphol Saranurak.
\newblock Deterministic algorithms for decremental shortest paths via layered
  core decomposition.
\newblock In {\em Proceedings of the 2021 ACM-SIAM Symposium on Discrete
  Algorithms (SODA)}, pages 2478--2496. SIAM, 2021.

\bibitem[Din06]{dinitz2006dinitz}
Yefim Dinitz.
\newblock Dinitz’algorithm: The original version and even’s version.
\newblock In {\em Theoretical computer science}, pages 218--240. Springer,
  2006.

\bibitem[GRST21]{goranci2021expander}
Gramoz Goranci, Harald R{\"a}cke, Thatchaphol Saranurak, and Zihan Tan.
\newblock The expander hierarchy and its applications to dynamic graph
  algorithms.
\newblock In {\em Proceedings of the 2021 ACM-SIAM Symposium on Discrete
  Algorithms (SODA)}, pages 2212--2228. SIAM, 2021.

\bibitem[HRW20]{henzinger2020local}
Monika Henzinger, Satish Rao, and Di~Wang.
\newblock Local flow partitioning for faster edge connectivity.
\newblock {\em SIAM Journal on Computing}, 49(1):1--36, 2020.

\bibitem[JS22]{jin2022fully}
Wenyu Jin and Xiaorui Sun.
\newblock Fully dynamic st edge connectivity in subpolynomial time.
\newblock In {\em 2021 IEEE 62nd Annual Symposium on Foundations of Computer
  Science (FOCS)}, pages 861--872. IEEE, 2022.

\bibitem[KLOS14]{kelner2014almost}
Jonathan~A Kelner, Yin~Tat Lee, Lorenzo Orecchia, and Aaron Sidford.
\newblock An almost-linear-time algorithm for approximate max flow in
  undirected graphs, and its multicommodity generalizations.
\newblock In {\em Proceedings of the twenty-fifth annual ACM-SIAM symposium on
  Discrete algorithms}, pages 217--226. SIAM, 2014.

\bibitem[KRV09]{khandekar2009graph}
Rohit Khandekar, Satish Rao, and Umesh Vazirani.
\newblock Graph partitioning using single commodity flows.
\newblock {\em Journal of the ACM (JACM)}, 56(4):1--15, 2009.

\bibitem[KT18]{kawarabayashi2018deterministic}
Ken-ichi Kawarabayashi and Mikkel Thorup.
\newblock Deterministic edge connectivity in near-linear time.
\newblock {\em Journal of the ACM (JACM)}, 66(1):1--50, 2018.

\bibitem[KVV04]{kannan2004clusterings}
Ravi Kannan, Santosh Vempala, and Adrian Vetta.
\newblock On clusterings: Good, bad and spectral.
\newblock {\em Journal of the ACM (JACM)}, 51(3):497--515, 2004.

\bibitem[{\L}{\k{a}}c13]{lkacki2013improved}
Jakub {\L}{\k{a}}cki.
\newblock Improved deterministic algorithms for decremental reachability and
  strongly connected components.
\newblock {\em ACM Transactions on Algorithms (TALG)}, 9(3):1--15, 2013.

\bibitem[Li21]{li2021deterministic}
Jason Li.
\newblock Deterministic mincut in almost-linear time.
\newblock In {\em Proceedings of the 53rd Annual ACM SIGACT Symposium on Theory
  of Computing}, pages 384--395, 2021.

\bibitem[Liu20]{liu2020vertex}
Yang~P Liu.
\newblock Vertex sparsification for edge connectivity in polynomial time.
\newblock {\em arXiv preprint arXiv:2011.15101}, 2020.

\bibitem[Lou10]{louis2010cut}
Anand Louis.
\newblock Cut-matching games on directed graphs.
\newblock {\em arXiv preprint arXiv:1010.1047}, 2010.

\bibitem[LSY19]{liu2019short}
Yang~P Liu, Sushant Sachdeva, and Zejun Yu.
\newblock Short cycles via low-diameter decompositions.
\newblock In {\em Proceedings of the Thirtieth Annual ACM-SIAM Symposium on
  Discrete Algorithms}, pages 2602--2615. SIAM, 2019.

\bibitem[NS17]{nanongkai2017dynamic1}
Danupon Nanongkai and Thatchaphol Saranurak.
\newblock Dynamic spanning forest with worst-case update time: adaptive, las
  vegas, and o (n1/2-$\varepsilon$)-time.
\newblock In {\em Proceedings of the 49th Annual ACM SIGACT Symposium on Theory
  of Computing}, pages 1122--1129, 2017.

\bibitem[NSWN17]{nanongkai2017dynamic}
Danupon Nanongkai, Thatchaphol Saranurak, and Christian Wulff-Nilsen.
\newblock Dynamic minimum spanning forest with subpolynomial worst-case update
  time.
\newblock In {\em 2017 IEEE 58th Annual Symposium on Foundations of Computer
  Science (FOCS)}, pages 950--961. IEEE, 2017.

\bibitem[OZ14]{orecchia2014flow}
Lorenzo Orecchia and Zeyuan~Allen Zhu.
\newblock Flow-based algorithms for local graph clustering.
\newblock In {\em Proceedings of the twenty-fifth annual ACM-SIAM symposium on
  Discrete algorithms}, pages 1267--1286. SIAM, 2014.

\bibitem[Pen16]{peng2016approximate}
Richard Peng.
\newblock Approximate undirected maximum flows in o (m polylog (n)) time.
\newblock In {\em Proceedings of the twenty-seventh annual ACM-SIAM symposium
  on Discrete algorithms}, pages 1862--1867. SIAM, 2016.

\bibitem[PY19]{parter2019optimal}
Merav Parter and Eylon Yogev.
\newblock Optimal short cycle decomposition in almost linear time.
\newblock In {\em 46th International Colloquium on Automata, Languages, and
  Programming (ICALP 2019)}. Schloss Dagstuhl-Leibniz-Zentrum fuer Informatik,
  2019.

\bibitem[Sar21]{saranurak2021simple}
Thatchaphol Saranurak.
\newblock A simple deterministic algorithm for edge connectivity.
\newblock In {\em Symposium on Simplicity in Algorithms (SOSA)}, pages 80--85.
  SIAM, 2021.

\bibitem[ST04]{spielman2004nearly}
Daniel~A Spielman and Shang-Hua Teng.
\newblock Nearly-linear time algorithms for graph partitioning, graph
  sparsification, and solving linear systems.
\newblock In {\em Proceedings of the thirty-sixth annual ACM symposium on
  Theory of computing}, pages 81--90, 2004.

\bibitem[SW19]{saranurak2019expander}
Thatchaphol Saranurak and Di~Wang.
\newblock Expander decomposition and pruning: Faster, stronger, and simpler.
\newblock In {\em Proceedings of the Thirtieth Annual ACM-SIAM Symposium on
  Discrete Algorithms}, pages 2616--2635. SIAM, 2019.

\bibitem[vdBLL{\etalchar{+}}21]{van2021minimum}
Jan van~den Brand, Yin~Tat Lee, Yang~P Liu, Thatchaphol Saranurak, Aaron
  Sidford, Zhao Song, and Di~Wang.
\newblock Minimum cost flows, mdps, and ℓ1-regression in nearly linear time
  for dense instances.
\newblock In {\em Proceedings of the 53rd Annual ACM SIGACT Symposium on Theory
  of Computing}, pages 859--869, 2021.

\bibitem[vdBLN{\etalchar{+}}20]{van2020bipartite}
Jan van~den Brand, Yin-Tat Lee, Danupon Nanongkai, Richard Peng, Thatchaphol
  Saranurak, Aaron Sidford, Zhao Song, and Di~Wang.
\newblock Bipartite matching in nearly-linear time on moderately dense graphs.
\newblock In {\em 2020 IEEE 61st Annual Symposium on Foundations of Computer
  Science (FOCS)}, pages 919--930. IEEE, 2020.

\bibitem[WN17]{wulff2017fully}
Christian Wulff-Nilsen.
\newblock Fully-dynamic minimum spanning forest with improved worst-case update
  time.
\newblock In {\em Proceedings of the 49th Annual ACM SIGACT Symposium on Theory
  of Computing}, pages 1130--1143, 2017.

\end{thebibliography}

\end{document}